\documentclass[11pt,english]{article}
\usepackage[english]{babel}
\usepackage[margin=1in]{geometry}
\usepackage{hyperref}
\hypersetup{
     colorlinks=true,
     linkcolor=blue,
     citecolor=blue,
 }
\usepackage{bm,xfrac,amsmath,amsthm,amssymb,soul, comment, xspace}
\usepackage{amsthm,tabularx,hyperref,amssymb,nicefrac,amsmath,multirow}
\usepackage{comment,amsfonts} 
\usepackage{epsfig} 
\usepackage{latexsym,nicefrac,bbm}
\usepackage{tabularx} 
\usepackage{hyperref} 

\usepackage{enumitem}
\usepackage{booktabs}
\usepackage{color,colortbl}

\newcommand{\nfrac}{\nicefrac}

\long\def\symbolfootnote[#1]#2{\begingroup%
\def\thefootnote{\fnsymbol{footnote}}\footnote[#1]{#2}\endgroup}

\newcommand{\CA}{\mbox{${\mathcal A}$}}
\newcommand{\CR}{\mbox{${\mathcal R}$}}

\newcommand{\CG}{\mbox{${\mathcal G}$}}
\newcommand{\CF}{\mbox{${\mathcal F}$}}
\newcommand{\CP}{\mbox{${\mathcal P}$}}
\newcommand{\CM}{\mbox{${\mathcal M}$}}

\newcommand{\zero}{\mbox{\boldmath $0$}}

\newcommand{\li}{{\lambda_i}}
\renewcommand{\tt}{\mbox{\boldmath $t$}}

\newcommand{\yy}{\mbox{\boldmath $y$}}
\newcommand{\zz}{\mbox{\boldmath $z$}}

\newcommand{\pp}{\mbox{\boldmath $p$}}

\newcommand{\pgamma}{\mbox{\boldmath $\gamma$}}
\newcommand{\plambda}{\mbox{\boldmath $\lambda$}}

\newcommand{\pdelta}{\mbox{\boldmath $\delta$}}

\newcommand{\xx}{\mbox{\boldmath $x$}}

\newcommand{\dd}{\mbox{\boldmath $d$}}
\renewcommand{\AA}{\mbox{\boldmath $A$}}
\renewcommand{\ss}{\mbox{\boldmath $s$}}

\newcommand{\CS}{\mathcal{S}}

\newcommand{\R}{\mathbb{R}}

\newcommand{\Rplus}{\mathbb{R}_+}

\newcommand{\ra}{\rightarrow}

\newcommand{\x}{\mbox{\boldmath $x$}}

\newcommand{\p}{\mbox{\boldmath $p$}}
\newcommand{\q}{\mbox{\boldmath $q$}}

\newcommand\Z{\rule{0pt}{3.2ex}}       
\newcommand\Y{\rule{0pt}{2.5ex}}       
\newcommand\B{\rule[-2.2ex]{0pt}{0pt}} 
\newcommand\M{\rule[-1.5ex]{0pt}{0pt}} 
\definecolor{Gray}{gray}{0.9}

\usepackage{float}

\newcommand{\ddelta}{\mbox{\boldmath $\delta$}}
\newcommand{\bbeta}{\mbox{\boldmath $\beta$}}
\newcommand{\ER}{\mbox{$\exists \mathbb{R}$}}

\newtheorem{theorem}{Theorem}[section]
\newtheorem{lemma}[theorem]{Lemma}

\newtheorem{corollary}[theorem]{Corollary}

\newtheorem{remark}[theorem]{Remark}

\newtheorem{definition}[theorem]{Definition}

\title{Settling the Complexity of Arrow-Debreu Markets under Leontief and PLC Utilities, using the Classes FIXP and $\ER$\thanks{Supported by NSF Grants CCF-1942321 (CAREER), CCF-1750436 (CAREER), and CCF-1815901. 
A preliminary version of this paper appeared at the 49th Symposium on Theory of Computing (STOC 2017)~\cite{GargMVY17}.}}
\author{Jugal Garg\thanks{University of Illinois at Urbana-Champaign. \texttt{jugal@illinois.edu}}
\and  Ruta Mehta\thanks{University of Illinois at Urbana-Champaign. \texttt{rutameht@illinois.edu}} 
\and Vijay V. Vazirani\thanks{University of California, Irvine. \texttt{vazirani@ics.uci.edu}}
\and  Sadra Yazdanbod\thanks{Google, New York. \texttt{sadra.yazdanbod@gmail.com}}}  

\begin{document}
\date{}
\maketitle

\begin{abstract}
This paper resolves two of the handful of remaining questions on the computability of market equilibria, a central theme within algorithmic game theory (AGT). Our results are as follows:

\begin{enumerate}
	\item We show FIXP-hardness of computing equilibria in Arrow-Debreu markets under Leontief utility functions, and Arrow-Debreu markets under linear utility functions and Leontief production sets. We note that these are the first FIXP-hardness results ever since the introduction of the class FIXP \cite{EY07} and the hardness of 3-Nash established therein.
	\item
	We note that for the problems stated above, the corresponding results showing membership in FIXP were established after imposing suitable sufficiency conditions to render the problems total, as is customary in economics. However, if all instances are under consideration, then in both cases we prove that the problem of deciding if a given instance admits an equilibrium is $\ER$-complete, where $\ER$ is the class of decision problems which can be reduced in polynomial time to Existential Theory of the Reals. 

\item
For Arrow-Debreu markets under Leontief utility functions and a constant number of agents, we give a polynomial time algorithm for computing an equilibrium. This settles part of an open problem of \cite{DK08}.
\end{enumerate}

We note that PLC utilities are about the most general utilities of interest in economics and several fundamental utility functions studied within AGT are special cases of it.
Several important problems, which have been shown to be in FIXP, are waiting for proofs of FIXP-hardness. In this context, our technique of reducing from 3-Nash to Multivariate Polynomial Equations and then to the problem is likely to be useful in the future.
\end{abstract}

\section{Introduction}

A central theme of algorithmic game theory (AGT), ever since its inception two decades ago, has been to determine the computational complexity of equilibrium problems. These come in two broad categories: game theoretic equilibria and market equilibria. Whereas the first category is dominated by the solution concept of Nash equilibrium, the second addresses a rather large collection of fundamental market models and utility functions.

This paper resolves two of the handful of remaining questions in the second category, namely establishing hardness of computing equilibria for Leontief and piecewise-linear concave (PLC) utility functions for the Arrow-Debreu market model, with and without production. In economics, it is customary to assume that utility functions are non-separable concave, not only because of their generality but also because they model the key property of decreasing marginal utilities. Since we wish to use a finite precision model of computation, we restrict attention to PLC utility functions; clearly, by making the pieces fine enough, we can obtain a good approximation to the original utility function.

Several utility functions studied over the years in AGT are special cases of PLC utilities, e.g., linear, additively separable piecewise-linear concave (SPLC) and Leontief. It should not come as a  surprise therefore that even within the two questions studied here, the picture is very rich. We describe this next, after stating some fundamental facts about the complexity of equilibrium problems. 

Equilibrium problems are {\em total} in that every {\em valid} instance admits a solution, i.e., an equilibrium. This is in stark contrast with NP-complete problems whose instances may not admit a solution and in fact simply deciding if a solution exists is NP-hard. Two novel classes have been defined for capturing the complexity of total problems, namely PPAD, defined in \cite{pap}, and FIXP, defined in \cite{EY07}. These two classes have captured the complexity of almost all prominent problems. A prerequisite for using the first class is that the problem must admit a solution using rational numbers; the second class is useful if the problem has instances having  only irrational, though algebraic, solutions. For instance, whereas {\em 2-Nash}, i.e., 2-player Nash equilibrium, always admits rational equilibria and is PPAD-complete \cite{DGP,CDT}, {\em $k$-Nash}, i.e., $k$-player Nash equilibrium, for $k \geq 3$, has instances with only irrational equilibria, and this problem is FIXP-complete \cite{EY07}. Interestingly enough, an irrational example for 3-Nash was given by Nash himself \cite{nash}.   

A basic difference between Nash equilibrium and market equilibrium is that whereas in the former {\em every} instance admits an equilibrium, in the latter this is not the case. Economists have rendered the latter total by giving suitable {\em sufficiency conditions}; any instance satisfying these conditions admits an equilibrium. As a consequence, for market equilibrium, two computational problems naturally arise: one dealing with the decision problem of determining if an instance admits an equilibrium and the other dealing with the complexity of solving the total problem, i.e., after imposing a suitable sufficiency condition.

Market equilibrium under Leontief utilities exhibits an intriguing set of possibilities. Andrew Mas-Colell gave an example having only irrational equilibria, even for the Fisher market model, which is a special case of the Arrow-Debreu model; this was reported in \cite{Eaves}. Despite this fact, equilibrium under Leontief utilities for the Fisher model can be computed efficiently to any desired accuracy, since it admits a convex program \cite{eisenberg}; furthermore, unlike other market equilibrium problems, this problem is total. In contrast, for the Arrow-Debreu model, this problem is PPAD-hard \cite{CSVY}; see Section \ref{sec.past} for details of this result.  

Observe that as a consequence of Mas-Colell's irrational example, we know that the right class for capturing the complexity of Leontief utilities for the Arrow-Debreu model should be FIXP and not PPAD. Membership in FIXP for Leontief utilities was established in \cite{yannakakis}. Subsequently, for PLC utilities and polyhedral production sets, membership in FIXP was established in \cite{GargMV16}. It should be noted that each of these results imposes a suitable sufficiency condition in order to render the problem total. 

In this paper, we prove FIXP-hardness for Arrow-Debreu markets under Leontief utility functions, and also for Arrow-Debreu markets under linear utility functions and Leontief production sets.  As corollaries, we obtain FIXP-hardness for Arrow-Debreu markets under PLC utilities and for Arrow-Debreu markets under linear utility functions and polyhedral production sets. We note that the class FIXP does not yield easily to proofs of membership and hardness. Indeed, ours are the first FIXP-hardness results after 3-Nash \cite{EY07}. As discussed in Section \ref{sec.discuss}, several prominent problems, which have already been shown to be in FIXP, are waiting for FIXP-hardness proofs.

As a consequence of the results stated above, the entire computational difficulty of Arrow-Debreu markets under PLC utility functions lies in the Leontief utility subcase. This is perhaps the most unexpected aspect of our result, since Leontief utilities are meant only for the case that goods are perfect complements, whereas PLC utilities are very general, capturing not only the cases when goods are complements and substitutes, but also arbitrary combinations of these and much more; see Section \ref{sec.ad} for definitions.

Next, we need to clarify a subtle fact about our hardness results; we will do so by drawing an analogy with NP-hardness results. Let $\Pi$ be a problem that is known to be in NP. By establishing NP-hardness of $\Pi$, one is accomplishing two tasks simultaneously. First, one is proving that if $\Pi$ is computationally easy, then so is every problem in NP. Second, $\Pi$  can be used for establishing  NP-hardness of other problems, by reducing from it. 

Our hardness of Leontief utilities accomplishes the first task but not the second. Since we reduce from 3-Nash, an algorithm for Leontief markets would yield an algorithm for 3-Nash. Let $S$ be the set of instances satisfying the chosen sufficiency condition for Leontief markets; clearly, all instances in $S$ admit equilibria. Observe that in order to accomplish the second task, our reduction would need to map {\em all} 3-Nash instances to instances in $S$. If a subset of 3-Nash is mapped outside $S$, then this total version of Leontief equilibrium may have consisted of ``easy'' instances only and hence our hardness proof did not accomplish the second task. It turns out that our reduction does not map all 3-Nash instances to instances in $S$, and we believe arranging this for a simply-stated, natural sufficiency condition is extremely non-trivial. At the same time, in the Remark appearing at the end of Section \ref{sec.res}, we do give a sufficiency condition under which both problems can be seen to be FIXP-complete; however, this condition does not have a simple description. On the positive side, as discussed in Section \ref{sec.discuss}, our paper does provide a novel avenue to establishing new FIXP-hardness results. 

Note that corresponding to each Leontief-based and PLC-based problem which was shown to be in FIXP, a suitable sufficiency condition was imposed to render the problem total. As discussed above, on removing the sufficiency condition, we obtain a new type of problem, namely of determining the complexity of the corresponding decision problem of whether a given instance admits an equilibrium. 

The second type of results in this paper address these problems. We show that all these problems are $\ER$-complete, where $\ER$ is the class of decision problems which can be reduced in polynomial time to Existential Theory of the Reals. In particular, this involves proving membership  of these problems in $\ER$ and showing $\ER$-hardness. 

Finally, our paper has a third type of result, which is of a positive nature. For Arrow-Debreu markets under PLC utilities, a polynomial time algorithm for finding an equilibrium was given in \cite{DK08}, provided the number of goods is a constant. They used the technique of algebraic cell decomposition using, in particular, ideas going back to \cite{Basu1995}. \cite{DK08} stated the open problem of obtaining a polynomial time algorithm for finding an equilibrium if the number of agents is a constant. We show how to handle the subcase of Leontief utilities, hence settling a part of this open problem.

\subsection{Previous results on computability of market equilibria}
\label{sec.past}

The first utility functions to be studied were linear. Once polynomial time algorithms were obtained for markets under such functions \cite{DPS,DPSV,DV1,JMS,GK,JainAD,Ye-AD,orlin,vegh,DM,DuanGM16,GargV19} and certain other cases \cite{CPV,JV,DK08,GKV,GargMSV15}, the next question was settling the complexity of Arrow-Debreu markets under additively separable piecewise-linear concave (SPLC) utility functions. This problem was shown to be complete for PPAD \cite{Chen.plc,VY}.  Also, when all instances are under consideration, the problem of deciding if a given SPLC market admits an equilibrium was shown to be NP-complete \cite{VY}.  The notion of SPLC production sets was defined in \cite{GargVaz} and Arrow-Debreu markets under such production sets and linear utility functions were shown to be PPAD-complete in the same paper. 

Previous computability results for Leontief utility functions were the following: In contrast to our result, Fisher markets under Leontief utilities admit a convex program \cite{Eisenberg61} and hence their equilibria can be approximated to any required degree in polynomial time \cite{BV,BSS}.  Arrow-Debreu markets under Leontief utilities were shown to be PPAD-hard \cite{CSVY} by reducing 2-Nash to a special case called ``pairing economy'' in which agents are partitioned into two sets, each agent brings one unit of a distinct good, and agents in each set desire the goods of the other set only. For this case, equilibria are rational. However, in general they are irrational for Leontief markets, as was shown by Mas-Colell and reported in \cite{Eaves}. We note that the two complexity classes PPAD and FIXP appear to be quite disparate -- whereas solutions to problems in the former are rational numbers, those to the latter are algebraic numbers. And whereas the former is contained in function classes NP $\cap$ co-NP, the latter lies somewhere between P and PSPACE, and is likely to be closer to the harder end of PSPACE \cite{yannakakis}.

Leontief utilities are a limiting case of constant elasticity of substitution (CES) utilities \cite{mas}. 
Finding an approximate equilibrium under CES was shown to be PPAD-complete~\cite{CPY}. An exact equilibrium, in this case, was shown to be in FIXP~\cite{CPY}, but the question of whether it is FIXP-hard is still unresolved.

\subsection{Technical contributions}\label{sec.tech}
Perhaps the most elementary way of stating the main technical part of our result is the following reduction, which we will denote by $\CR$: Given a set $\CS$ of simultaneous multivariate polynomial equations in which the variables are constrained to be in a closed bounded region in the positive orthant, we construct an Arrow-Debreu market with Leontief utilities, say $\CM$, which has one good corresponding to each variable in $\CS$. We prove that the equilibria of $\CM$, when projected onto prices of these latter goods, are in one-to-one correspondence with the set of solutions of the polynomials. This reduction, together with the fact that 3-Nash is FIXP-complete \cite{EY07} and that 3-Nash can be reduced to such a system $\CS$, yield FIXP-hardness for the Leontief case. 

We now describe the difficulties encountered in obtaining reduction $\CR$ and the ideas needed to overcome them; this should also help explain why FIXP-hardness of Leontief (and PLC) markets was a long-standing open problem. For this purpose, it will be instructive to draw a comparison between reduction $\CR$ and the reduction from 2-Nash to SPLC markets given in \cite{Chen.plc}.  At the outset, observe the latter is only dealing with linear functions of variables\footnote{Since the payoff of the row player from a given strategy is a linear function of the variables denoting the probabilities played by the column player.} and hence is much easier than the former.

Both reductions create one market with numerous agents and goods, and the amount of each good desired by an agent gets determined only after the prices are set. Yet, at the desired prices, corresponding to solutions to the problem reduced from, the supply of each good needs to be exactly equal to its demand. In the latter reduction, the relatively constrained utility functions give a lot more ``control'' on the optimal bundles of agents.  Indeed, it is possible to create one large market with many agents and many goods and still argue how much of each good is consumed by each agent at equilibrium. 

We do not see a way of carrying out similar arguments when all agents have Leontief utility functions. The key idea that led to our reduction was to create several modular units within the large market and ensure that each unit would have a very simple and precise interaction with the rest of the market. Leontief utilities, which seemed hard to manage, in fact enabled this in a very natural manner as described below.
\smallskip

{\bf Closed submarket:} A closed submarket is a set $S$ of agents satisfying the following: At every equilibrium of the complete market, 
the union of initial endowments of all agents in $S$ exactly equals the union of optimal bundles of all these agents. 
\smallskip

Observe that the agents in $S$ will not be sequestered in any way --- they are free to choose their optimal bundles from all
the goods available.  Yet, we will show that at equilibrium prices, they will only be exchanging goods among themselves. 
We note that the proof of PPAD-hardness for Nash equilibrium computation also uses small game gadgets to accomplish arithmetic
operations of addition, multiplication and comparison \cite{DGP,CDT}; however, since there variables are captured through strategies of
different players, these gadgets did not interfere. In our case, the primary challenge is to prevent flow of goods across gadgets and 
to ensure desired price dependencies even when same goods are used across gadgets. We achieve this through the notion of closed submarkets.

These closed submarkets enable us to ensure that variables denoting prices of goods satisfy specified arithmetic relations. The latter are
linear function and product; we show that these two arithmetic relations suffice to encode any polynomial equation. Under linear functions, we want
that $p_a = B p_b + C p_c + D$, where $B, C$ and $D$ are constants. 

Under product, we want that $p_a = p_b \cdot p_c$. Designing this closed submarket, say $\CM$, requires several ideas, which we now describe. 
$\CM$ has an agent $i$ whose initial endowment is one unit of good $a$ and she desires only good $c$. We will ensure that the {\em amount}
of good $c$ leftover, after all other agents in the submarket consume what they want, is exactly $p_b$, i.e., the {\em price} of good $b$. At equilibrium,
$i$ must consume all the leftover good $c$, whose total cost is $p_b \cdot p_c$. Therefore the price of her initial endowment, i.e., one unit of good $a$, must be 
$p_b \cdot p_c$, hence establishing the required product relation.
The tricky part is ensuring that exactly $p_b$ amount of good $c$ is leftover, without knowing what $p_b$ will be at equilibrium. This is non-trivial, and
this submarket needs to have several goods and agents in addition to the ones mentioned above. 

Once reduction $\CR$ is established,
FIXP-hardness follows from the straightforward observation that a 3-Nash instance can be encoded via polynomials, where each variable, which represents the probability of playing a certain strategy, is constrained in the interval $[0, 1]$. 

To get $\ER$-hardness, we appeal to the
result of \cite{SS} that checking if a 3-Nash instance has a solution in a ball of radius half in $l_{\infty}$-norm is $\ER$-hard;
this entails constraining the variables to be  in the interval $[0, 1/2]$. By Nash's theorem, in the former case, the market will admit an equilibrium 
and in the latter case, it will admit an equilibrium iff the 3-Nash instance has a solution in the ball of radius half in $l_{\infty}$-norm.
Membership in $\ER$ follows by essentially showing a reduction in the reverse direction: given a Leontief market, we obtain a set of simultaneous multivariate polynomial equations whose roots capture its equilibria. 

Our final result gives a polynomial time algorithm for computing an equilibrium for Arrow-Debreu exchange markets under Leontief utility functions provided the number of agents is a constant, say $d$. Using the property that equilibrium allocation of an agent can be written in terms of her equilibrium utility, we show that if equilibrium exists, then there is one where the number of goods with positive prices is at most $d$. Next, we iterate over all subsets of size $d$ of goods, and for each set we reduce the problem of checking existence of equilibrium to checking feasibility of a set of polynomial inequalities in $2d$ dimension. Since this can be done in polynomial time \cite{bpr,bpr-book}, we get a polynomial time algorithm.

\medskip

\noindent {\bf Organization.} 
In Section \ref{sec.ad} we define the Arrow-Debreu market model, and the relevant utility and production functions. The definition of 3-player Nash equilibrium problem, and the complexity classes FIXP and $\ER$ are given in Section \ref{sec.nash}. Section \ref{sec.mred} contains our main result where we show FIXP-hardness of computing an equilibrium in Leontief exchange markets. In Section \ref{sec.eetr}, we show that checking existence of an equilibrium in PLC markets is in $\ER$. 
Section \ref{sec.const} describes a polynomial time algorithm for computing an equilibrium in Leontief exchange markets provided the number of agents is a constant. 
\medskip

\noindent {\bf Notation.} Capital letters denote matrices of constants, like $W$; bold
lower case letters denote vector of variables, like $\xx,\yy$; and calligraphic capital letters denote sets like $\CA, \CG$. We use $[n]$ to denote the set $\{1,\dots,n\}$. Given an $n$-dimensional vector $\xx$ and a number $r \in \R$, by $\xx \le r$, we mean $\forall i\in[n],\ x_i\le r$.
\medskip

\section{The Arrow-Debreu Market Model}
\label{sec.ad}

An Arrow-Debreu (AD) market \cite{AD} consists of a set $\CG$ of divisible goods, a set $\CA$ of agents and a set $\CF$ of firms. Let $g$ denote the number of goods in the market. 

The production capabilities of a firm is defined by a convex set of production schedules and each firm wants to use a (optimal)
production schedule that maximizes its profit $-$ money earned from the production minus the money spent on the raw materials. 
Firms are owned by agents: $\Theta_{if}$ is the profit share of agent $i$ in firm $f$ such that $\forall f\in \CF,\
\sum_{i\in\CA} \Theta_{if}=1$.
Each agent $i$ has a utility function $U_i: \Rplus^g \ra \Rplus$ over bundles, and comes with an initial endowment of goods; $W_{ij}$ is
amount of good $j$ with agent $i$. 
Each agent wants to buy a (optimal) bundle of goods that maximizes her utility to the extent allowed by her earned money --
from initial endowment and profit shares in the firms. 

Given prices of goods, if there is an assignment of optimal production schedule to each firm and optimal affordable bundle
to each agent so that there is neither deficiency nor surplus of any good, then such prices are called {\em market
clearing} or {\em market equilibrium} prices; we note that a zero-priced good is allowed to be in surplus. 
The market equilibrium problem is to find such prices when they exist. In a celebrated
result, Arrow and Debreu \cite{AD} proved that market equilibrium always exists under some mild conditions, however the
proof is non-constructive and uses heavy machinery of Kakutani fixed point theorem. We note that an arbitrary market may not admit an
equilibrium.

A well studied restriction of Arrow-Debreu model is {\bf exchange economy}, i.e., markets without production firms.
To work under finite precision it is customary to assume that utility functions are piecewise-linear concave (PLC) and
production sets are polyhedral.

\subsection{Piecewise-linear concave (PLC) utility function}
\label{sec.util}

As stated earlier, agent $i$'s utility function is $U_i: \Rplus^g \ra \Rplus$ over bundle of goods.
These functions are said to be piecewise-linear concave (PLC) if at bundle
$\xx_i=(x_{i1},\dots,x_{ig})$ it is given by:
\[ U_i(\xx_i) = \min_{k} \{\sum_j U^k_{ij}x_{ij} + T^k_i\},\]
where $U^k_{ij}$'s and $T^k_i$'s are given non-negative rational numbers. Since the agent gets zero utility when she gets nothing, we
have $U_i(\zero)=0$, and therefore at least one $T^k_i$ is zero.

\subsubsection{Leontief utility function}
\label{sec.leon}

The Leontief utility function is a special subclass of PLC, where each good is required in a fixed proportion.
Formally, it is given by: \[ U_i(\xx_i)=\min_{j \in \CG: A_{ij}>0} \left\{\frac{x_{ij}}{A_{ij}}\right\},\] where $A_{ij}$'s are
non-negative numbers. In other words the agent wants good $j$ in $A_{ij}$ proportion; $A_{ij}=0$ implies good $j$
is not desired by the agent. Clearly, the agent has to spend $\sum_j A_{ij}p_j$ amount of money to get one unit of utility. Thus, the 
optimal bundle satisfies the following condition.
\begin{equation}\label{eq.leon}
\forall j \in \CG,\ x_{ij}=\beta_i A_{ij},\ \ \ \ \mbox{where}\ \ \beta_i = \frac{ \sum_{j\in \CG} W_{ij} p_j }{\sum_{j\in \CG} A_{ij}p_j} \enspace .
\end{equation}

\subsection{PLC production sets}\label{sec.proddef}
A firm can produce a set of goods using another set of goods as raw materials, and these two sets are assumed to be disjoint. 
Let $\CP_f \in \R^g$ be the set of production schedules for firm $f$, then if it can produce a bundle $\xx^s$ using a bundle
$\xx^r$ then $\xx^s - \xx^r \in \CP_f$. The set is assumed to be downward closed, contains the origin $\zero$, no vector is
strictly positive ({\em no production out of nothing}) and is polyhedral. We call these {\em PLC production sets}.  

Let $\CS_f$ denote the set of goods that can be produced by firm $f$ and $\CR_f$ be the set of goods it can use as raw material such
that $\CS_f\cap \CR_f=\emptyset$. 
A PLC production set of firm $f$ can be described as follows, where $x^s_{fj}$ and $x^r_{fj}$ denote the amount of good $j$ produced
and used respectively. 
\begin{equation}
\begin{aligned}
\CP_f = \Big\{(\xx^s-\xx^r) \in \R^g \ | \ \sum_{j\in \CS_f} D^k_{fj}x^{s}_{fj} & \le \sum_{j \in \CR_f} C^k_{fj}x^{r}_{fj} + T^k_{f},\ \forall k;\\  
 \xx^s & \ge 0; \ \ x^s_{fj}=0,\ \forall j \notin \CS_f;\\ 
 \xx^r & \ge 0; \ \ x^r_{fj}=0,\ \forall j \notin \CR_f\Big\}, 
\end{aligned}
\end{equation}
where $D^k_{fj}$'s, $C^k_{fj}$'s and $T^k_{f}$'s are given non-negative rational numbers. 
Since there is no production if no raw material is consumed, it should be the case that for some $k$, $T^k_{f}=0$.

\subsubsection{Leontief production}
The Leontief production is a special subclass of PLC production sets, where a firm $f$ produces a single good $a$ using a
subset of the rest of the goods as raw materials. To produce one unit of $a$, it requires $D_{fj}$ units of good $j$, i.e.,   
\[x^s_{fa}=\min_{j \neq a}\left\{\frac{x^r_{fj}}{D_{fj}}\right\}\enspace .\]\smallskip

\section{$3$-player Nash Equilibrium ($3$-Nash)}\label{sec.nash} 
In this section we describe 3-player finite games and characterize their Nash equilibria. 
Given a $3$-player finite game, let the set of strategies of player $p\in\{1,2,3\}$ be denoted by $\CS_p$. 
Let $\CS=\CS_1\times\CS_2\times\CS_3$. Without loss of generality, we assume that
$|\CS_1|=|\CS_2|=|\CS_3|=n_s$. Such a game can be represented by $n_s\times n_s\times n_s$-dimensional 
tensors $A_1, A_2$ and $A_3$ representing payoffs of first, second and third player respectively.  
If players play $\ss=(s_1,s_2,s_3)\in \CS$, then the payoffs are $A_1(\ss)$, $A_2(\ss)$ and $A_3(\ss)$
respectively.

Since players may randomize among their strategies, let $\Delta_p$ denote the probability distribution over set $\CS_p,\
\forall p\in\{1,2,3\}$ (the set of mixed-strategies for player $p$), and let $\Delta=\Delta_1\times\Delta_2\times \Delta_3$. Given a
mixed-strategy profile $\zz=(\zz_1,\zz_2,\zz_3) \in \Delta$, let $z_{ps}$ denote the probability with which
player $p$ plays strategy $s \in \CS_p$, and let $\zz_{-p}$ be the strategy profile of all the players at $\zz$ except $p$.
For player $p\in\{1,2,3\}$ the total payoff and payoff from strategy $s \in \CS_p$ at $\zz$ are respectively,
\[
\pi_p(\zz)=\sum_{\ss \in \CS} A_p(\ss) z_{1s_1} z_{2s_2} z_{3s_3}\ \ \ \mbox{ and }\ \ \ \pi_p(s,\zz_{-p})=\sum_{\tt \in \CS_{-p}}
A_p(s,\tt) \Pi_{q\neq p} z_{qt_{q}}\enspace . 
\]

\begin{definition}[Nash (1951) \cite{nash}]
A mixed-strategy profile $\zz \in \Delta$ is a Nash equilibrium (NE) if no player gains by deviating unilaterally. Formally, $\forall p\ \ \pi_p(\zz) \ge \pi_p(\zz',\zz_{-p}), \forall \zz' \in \Delta_p$.
\end{definition}

In 1951, Nash \cite{nash} proved existence of an equilibrium in a finite game using Brouwer fixed-point
theorem, which is highly non-constructive. 
Despite many efforts over the years, no efficient methods are obtained to compute a NE of finite games.
Next we give a characterization of NE through multivariate polynomials and discuss its complexity.
\medskip

\noindent{\bf NE Characterization.} It is easy to see that in order to maximize the expected payoff, only {\em best moves}
should be played with a non-zero probability; by {\em best moves} we mean the moves fetching maximum payoff.
Formally,
\begin{equation}\label{eq.nash}
\forall p,\ \forall s \in S_p,\   z_{ps} > 0 \Rightarrow \pi_p(s,\zz_{-p}) = \delta_p,\ \mbox{ where }\ \delta_p=\max_{s' \in S_p} \pi_p(s',\zz_{-p})\enspace . 
\end{equation}

\noindent{\em Assumption.} Since scaling all the co-ordinates of $A_p$'s with a positive number or adding a constant to
them does not change the set of Nash equilibria of the game $\AA=(A_1,A_2,A_3)$, without loss of generality we assume that all the
co-ordinates of each of $A_p$ are in the interval $[0, 1]$. 
Using (\ref{eq.nash}), we can define the following system of multivariate polynomials, where variables $z_{ps}$'s capture
strategies, $\delta_p$ captures payoff of player $p$, and $\beta_{ps}$ are slack variables:
\begin{equation}\label{eq.fne}
\begin{array}{ll}
F_{NE}(\AA): & \begin{array}{ll}
\forall p, & \sum_{s \in S_p} z_{ps}=1\M \\
\forall p,\ s \in S_p, & \pi_p(s,\zz_{-p}) + \beta_{ps} = \delta_p\ \mbox{ and }\ z_{ps}\beta_{ps}=0\M \\
\forall p,\ s \in S_p, & 0 \le z_{ps}\le 1,\ \ \ 0 \le \beta_{ps}\le 1,\ \ \ 0 \le \delta_p \le 1\M
\end{array}\enspace . 
\end{array}
\end{equation}

\begin{lemma}\label{lem.nash}
Nash equilibria of $\AA$ are exactly the solutions of system $F_{NE}(\AA)$, projected onto $\zz$.
\end{lemma}
\begin{proof}
It is easy to check that NE of $\AA$ gives a solution of $F_{NE}(\AA)$ using (\ref{eq.nash}); the upper bounds on variables
of $F_{NE}(\AA)$ holds because all the entries in $\AA$ are in the interval $[0, 1]$. Similarly, given a solution
$(\zz,\bbeta,\ddelta)$ of $F_{NE}(\AA)$, the first condition ensures that $\zz \in \Delta$. 
The two parts of the second condition imply that $\zz$ satisfies (\ref{eq.nash}) and therefore is a NE of game $\AA$.
\end{proof}

Let $3$-Nash denote the problem of computing Nash equilibrium of a $3$-player game. Next we describe complexity classes FIXP
and $\ER$, and their relation with $3$-Nash.

\subsection{The class FIXP}\label{sec.fixp}

The class FIXP, which was introduced in \cite{EY07}, consists of problems that can be cast, in polynomial time, as the problem of computing a fixed point of a continuous algebraic function.

 Basic complexity classes, such as P, NP, NC and \#P, are defined via machine models. In the same vein, in the case of FIXP, an algebraic circuit plays the role of ``machine model''. Furthermore, the circuit must use the standard arithmetic operations of $+$, $-$, $*$, $/$, $\min$ and $\max$. A total problem is in FIXP if there is a polynomial time algorithm which given an instance $I$ of the problem, outputs such a circuit. Under the restrictions imposed on the circuit, the function defined by this ``machine'' is guaranteed to be continuous, and hence by Brouwer's fixed point theorem, must have at least one fixed point. We further want that each fixed point must correspond to a solution of the total problem.

The problem of computing an equilibrium for a $k$-player game, for $k \geq 3$, game is FIXP-complete \cite{EY07}; in fact, this is the quintessential FIXP-complete problem. In contrast, a two-player bimatrix game always admits rational equilibria and computing it is PPAD-complete \cite{DGP,CDT}. The two classes, PPAD and FIXP seem to be rather disparate: whereas PPAD is contained in function classes NP $\cap$ co-NP, the class FIXP lies somewhere between P and PSPACE, and is likely to be closer to the harder end of PSPACE.

\medskip

\noindent{\bf Reduction requirements:} 
A reduction from problem $A$ to problem $B$ consists of two polynomial-time computable functions: a function $f$ that maps
an instance $I$ of $A$ to an instance $f(I)$ of $B$, and another function $g$ that maps a solution $\yy$ of $f(I)$ to a
solution $\xx$ of $I$. If $x_i=g_i(\yy)$, then $g_i(\yy)=a_i y_j +b_i$, for some $j$, where $a_i$ and $b_i$ are
polynomial-size rational numbers; every coordinate of $\xx$ is a linear function of one coordinate of $\yy$.
\medskip

\begin{theorem}\label{thm.fixp}\cite{EY07}
Given a $3$-player game $\AA=(A_1,A_2,A_3)$, computing its NE is FIXP-complete.  
\end{theorem}

\subsection{Existential Theory of Reals}\label{sec.etr}
The class $\ER$ was defined to capture the decision problems arising in {\em existential theory of reals}~\cite{SS}. An instance $I$ of $\ER$ consists of a sentence of the form,
\[
   (\exists x_1, \dots , x_n)\phi(x_1, \dots , x_n), 
\]
where $\phi$ is a quantifier-free $(\land,\lor,\lnot)$-Boolean formula over the predicates (sentences) defined by signature $\{0, 1, -1, +, *, <, \le, =\}$ over variables that take real values. The question is whether the sentence is true. Following is an example of such an instance,
\[
\exists(x_1, x_2),\ \ (x_1^4+x_1^3x_2^2-3x_2^3+1=0 \ \land\ x_1x_2 \ge 3)\ \ \lor \ \ (x_2^4 - 3x_1 < 6) \enspace . 
\]

\cite{SS} showed that disallowing the operation of $<$ does not change the class $\ER$.  The size of the problem is $n+size(\phi)$, where $n$ is the number of variables and
$size(\phi)$ is the minimum number of signatures needed to represent $\phi$ (we refer the readers to
\cite{SS} for detailed description of $\ER$, and its relation with other classes like PSPACE). 
\cite{SS} showed the following result; the first result on the complexity of a {\em decision} version of $3$-Nash.

\begin{definition}[Decision $3$-Nash]\label{def.d3n}
Decision $3$-Nash is the problem of checking if a given $3$-player game $\AA$ admits a Nash equilibrium $\zz$ such that $\zz \le 0.5$. 
\end{definition}

\begin{theorem}\label{thm.etr}\cite{SS}
Decision $3$-Nash is $\ER$-complete.
\end{theorem}

Note that changing the upper bound on all $z_{ps}$s from $1$ to $0.5$ in $F_{NE}(\AA)$ (\ref{eq.fne}), exactly captures the NE with $\zz \le 0.5$. Thus {\em Decision $3$-Nash} can be reduced to checking if such a system of polynomial equalities admits a solution. Next we show a construction of Leontief exchange markets to exactly capture the solutions of a system of multivariate polynomials, similar to that of $F_{NE}(\AA)$, at its equilibria.

\section{Multivariate Polynomials to Leontief Exchange Market}\label{sec.mred}
Consider the following system of $m$ multivariate polynomials on $n$ variables $\zz=(z_1,\dots,z_n)$,
\begin{equation}\label{eq.F}
F: \{f_i(\zz)=0,\forall i \in[m]; \ \ \ L_j \le z_j \le U_j,\ \forall j \in[n]\},\ \ \mbox{where $L_j,U_j \ge 0$} \enspace . 
\end{equation}

The coefficients of $f_i$'s, and the upper and lower bounds $U_j$'s and $L_j$'s are assumed to be rational numbers.
In this section, we show that solutions of $F$ can be captured as equilibrium prices of an exchange market with
Leontief utility functions (see Section \ref{sec.leon} for definition). The problems of $3$-Nash and Decision $3$-Nash (see Definition
\ref{def.d3n}) can be characterized by a set similar to (\ref{eq.F}) (see (\ref{eq.fne}) in Section \ref{sec.nash}), in turn we obtain
FIXP and $\ER$ hardness results for Leontief markets, and in turn PLC markets, from the corresponding hardness of $3$-Nash \cite{EY07,SS}.

Polynomial $f_i$ is represented as sum of monomials, and a monomial $\alpha z_1^{d_1}\dots z_n^{d_n}$ is represented by
tuple $(\alpha,d_1,...,d_n)$; here coefficient $\alpha$ is a rational number. Let $\CM_{f_i}$ denote the set of monomials of
$f_i$, and $size[f_i]=\sum_{(\alpha,\dd)\in \CM_{f_i}} size(\alpha,\dd)$, where $size(r)$ for a rational number $r$ is the minimum
number of bits needed to represent its numerator and denominator. Degree of $f_i$ is
$deg(f_i)=\max_{(\alpha,\dd)\in \CM_{f_i}} \sum_j d_j$. The size of $F$, denoted by $size[F]$, is $m+n+\sum_j
(size(U_j)+size(L_j)) +\sum_i (deg(f_i) + size[f_i])$. Given a system $F$, next we construct an exchange market in time polynomial in
$size[F]$, whose equilibria correspond to solutions of $F$.

\subsection{Preprocessing}\label{sec.pre}
First, we transform $F$ into a polynomial sized equivalent system that uses only the following three basic operations on {\em non-negative} variables.
\begin{equation}\label{eq.basic}
\begin{array}{ll}
(EQ.)& z_a=z_b\\
(LIN.)& z_a = Bz_b + C z_c +D,\ \mbox{ where }\ 
B,C,D \ge 0\\ 
(QD.)& z_a=z_b*z_c 
\end{array}\enspace . 
\end{equation}

\begin{remark}
We note that even though (EQ.) is a special case of (LIN.), we consider it separately in order to convey the main ideas. 
\end{remark}

Next we illustrate how to capture $f_i$'s using these basic operations through an example.  Consider a polynomial
\[ 4z_1^2z_2 + 3z_1z_2  - z_1 - 2=0\enspace . \]

First, move all monomials with negative coefficients to right hand side of the equality, so that all coefficients 
become positive, 
\[4z_1^2z_2 + 3z_1z_2  =  z_1 + 2 \enspace . \]

Second, capture every monomials, with degree more than one, using basic operations:
\[
\begin{array}{lcl}
z_{a_1}=z_1^2z_2 &\ \  \equiv\ \  &  z_{a_2}=z_1*z_1,\ z_{a_1}= z_{a_2}*z_2 \\ 
z_{b_1}=z_1z_2 &\ \  \equiv\ \  &  z_{b_1}=z_1*z_2 
\end{array}\enspace . 
\]

Third, capture the equality $4z_{a_1}+3z_{b_1}=z_1+2$ using $(LIN.)$ and $(EQ.)$:
\[
\begin{array}{lcl}
4z_{a_1}+3z_{b_1}=z_1+2 &\ \ 
\equiv\ \  & 
z_{e_1}=4z_{a_1}+3z_{b_1}, \ z_{f_1}=z_1+2,\ z_{e_1}=z_{f_1}
\end{array}\enspace . 
\]

Finally, combine all of the above to represent $f_i$ as follows: 

\begin{equation}\label{eq.eg}
\begin{array}{lcl}
4z_1^2z_2 + 3z_1z_2 - z_1 - 2=0 &
\equiv & 
\begin{array}{c}
z_{a_2}=z_1*z_1,\ z_{a_1}=z_{a_2}*z_2 \\ 
z_{b_1}=z_1*z_2,\ z_{e_1}=4z_{a_1}+3z_{b_1}\\ z_{f_1}=z_1+2,\ z_{e_1}=z_{f_1} \\
\end{array}\enspace . 
\end{array}
\end{equation}

Since inequalities have to be captured through equalities with non-negative variables, the inequalities of (\ref{eq.F}) have
to be transformed as follows:
\begin{equation}\label{eq.F'}
\begin{array}{lll}
& \forall j \in [n], & z_j = s^l_j + L_j,\ \ \ z_j+s^u_j =U_j, \ \ \ z_j, s^l_j, s^u_j \ge 0\enspace . 
\end{array}
\end{equation}

Let $R(F)$ be a reformulation of $F$, using transformation similar to (\ref{eq.eg}) for each $f_i$, and that of (\ref{eq.F'})
for each inequality. All the variables in $R(F)$ are constrained to be non-negative.  

In order to construct $R(F)$ from $F$, we need to introduce many auxiliary variables (as was done in (\ref{eq.eg})).
Let the number of variables in $R(F)$ be $N$, and out of these let $z_1,\dots,z_n$ be the original set of variables
of $F$ (\ref{eq.F}). 
Given a system $R(F)$ of equalities, we will construct an exchange market $\CM$, such that the value of each variable $z_j,\ j \in[N]$ 
is captured as price $p_j$ of good $G_j$ in $\CM$. Further, we make sure that these prices satisfy all the relations in
$R(F)$ at every equilibrium of $\CM$.

\subsection{Ensuring scale invariance}
\label{sec.scale}

Since equilibrium prices of an exchange market are scale invariant, the relations that these prices satisfy have to be scale invariant
too. However, note that in (\ref{eq.basic}) $(LIN.)$ and $(QD.)$ are not scale invariant. To handle this we introduce a special good
$G_s$, such that when its price $p_s$ is set to $1$ we get back the original system.
\begin{equation}\label{eq.trans}
\begin{array}{ll}
(EQ.)& p_a=p_b \M\\
(LIN.)& p_a = B p_b + C p_c +D p_s,\ \mbox{ where }\ B,C,D \ge 0 \M\\ 
(QD.)& p_a=\frac{p_b*p_c}{p_s}
\end{array}\enspace . 
\end{equation}

Let $R'(F)$ be a system of equalities after applying the transformation of (\ref{eq.trans}) to $R(F)$. Note that, $R'(F)$ has exactly
one extra variable than $R(F)$, namely $p_s$, and solutions of $R'(F)$ with $p_s=1$ are exactly the solutions of $R(F)$. 

Let the size of $R'(F)$ be (\# variables + \# relations in $R'(F)$ + $size(B,C,D)$ in each of (LIN.)-type relations).
Recall that $\CM_{f_i}$ denote the set of monomials in polynomial $f_i$.
To bound the values at a solution of $R'(F)$, define

\[
\begin{array}{lcl}
H= M_{max} U_{max}^d + 1, & \mbox{ where } & \begin{array}{l}
d=\max_{f_i} deg(f_i), \ \ \ M_{max} = \max_i |\CM_{f_i}|,\\
U_{max}=\max\big\{\max_j U_j, \displaystyle\max_{\mbox{\small $f_i,
(\alpha,\dd) \in \CM_{f_i}$}} |\alpha|\big\}\enspace .\Z 
\end{array}
\end{array}
\]

\begin{lemma}\label{lem.RF}
$size[R'(F)]=poly(size[F])$. Vector $\pp$ is a non-negative solution of $R'(F)$ with $p_s=1$ iff $z_j=p_j,\ \forall j\in [n]$ is a solution of
$F$. Further, $p_j \le H,\ \forall j \in [N]$.
\end{lemma}
\begin{proof}
For the first part, it is enough to bound $size[R(F)]$. Note that the number of auxiliary variables added to $R(F)$ to
construct a monomial of $f_i$ is at most its degree. Further, to construct the expression of $f_i$ from these, the number of extra
variables needed is at most the number of monomials. Further, the coefficients of
$(LIN.)$ type relations are coefficients of the monomials of $f_i$'s. Thus, we get that $size[R'(F)]=O(\sum_{i \in [m]}
deg(f_i)^2 size[f_i]+\sum_{j \in[n]} size(U_j)+size(L_j))=poly(size[F])$.

The second part follows by construction. For the third part note that $(p_1,\dots,p_N)$ is a solution of $R(F)$. 
Since variables of the original system $F$ is upper bounded by $U_j$'s, it is easy to see that the
maximum value of any variable in a non-negative solution of $R(F)$ is at most $H$. 
\end{proof}

Next, we construct a market whose equilibria satisfy all the relations of $R'(F)$, and has $p_s>0$.

\subsection{Market construction}\label{sec.red}
In this section, we construct 
market $\CM$ consisting of goods $G_1,\dots,G_N$ and $G_s$, such that the prices $p_1,\dots,p_N$ and $p_s$, satisfy all the relations
of $R'(F)$ at equilibrium.

First, we want price of $G_s$ to be always non-zero at equilibrium. To ensure this we add the following agent to market $\CM$.
Recall that $W_{ij}$ is the amount of good $G_j$ agent $A_i$ brings to the market, $U_i:\Rplus^g\ra\Rplus$ is the utility
function of $A_i$, and $\xx_i$ denotes the bundle of goods consumed by her.
\begin{equation}\label{eq.sms}
A_s:\  W_{ss}=1,\ W_{sj}=0,\ \forall j \in[N];\ \ \  U_s(\xx_s)=x_{ss}\enspace . 
\end{equation}

\begin{lemma}\label{lem.sms}
At every equilibrium of market $\CM$, we have $p_s>0$, and $x_{ss}=W_{ss}$.
\end{lemma}
\begin{proof}
At an equilibrium if $p_s=0$, then agent $A_s$ will demand infinite amount of good $s$, a contradiction.
The second part follows using the fact that at any given prices, $A_s$ wants to buy only good $s$, and has exactly $W_{ss}p_s$
amount of money to spend. 
\end{proof}

Since a price $p_j$ may be used in multiple relations of $R'(F)$, the corresponding good has to be used in many different gadgets. 
When we combine all these gadgets to form market $\CM$, 
the biggest challenge is to analyze the flow of goods among these gadgets at equilibrium. We overcome this all together by forming {\em
closed submarket} for each gadget. 

\begin{definition}[submarket]
A submarket $\widetilde{\CM}$ of a market $\CM$ consists of a subset of agents and goods such that the
endowment and utility functions of agents in $\widetilde{\CM}$ and the production functions of firms in $\widetilde{\CM}$
are defined over goods only in $\widetilde{\CM}$. 
\end{definition} 

\begin{definition}[closed submarket]\label{def.csm}
A submarket $\widetilde{\CM}$ of a market $\CM$ is said to be closed if at every equilibrium of the entire market $\CM$, the
submarket $\widetilde{\CM}$ is locally at equilibrium, i.e., its total demand equals its total supply. The total demand of
$\widetilde{\CM}$ is the sum of demands of agents in $\widetilde{\CM}$ and its total supply is the sum of initial endowments
of agents in $\widetilde{\CM}$.  
\end{definition}

In other words, $\widetilde{\CM}$ does not interfere with the rest of market in terms of supply and demand, even if some goods in
$\widetilde{\CM}$ are used outside as well. Note that the market of (\ref{eq.sms}) is a closed submarket (due to Lemma \ref{lem.sms})
with only one agent and one good, namely $A_s$ and $G_s$ respectively.
We will see that the submarket $\widetilde{\CM}$ establishing a relation of type $(EQ.)$, $(LIN.)$ and $(QD.)$ has
a set of {\em exclusive} goods used only in $\widetilde{\CM}$, in order to achieve the closed property. 
Before describing construction of closed submarkets for more involved relations, we first describe it for a simple and important {\em
equality} relation. Furthermore, we will use equality to construct closed markets for $(QD.)$. 

Let there be $K$ relations in $R'(F)$, numbered from $1$ to $K$, and let $\CM_r$ denote the closed submarket establishing relation
$r \in[K]$. 
\subsubsection{Submarket for relation $(EQ.)\ p_a=p_b$}\label{sec.eq}

The gadget for $(EQ.)$ consists of two agents with Leontief utility functions, as given in Table
\ref{tab.eq}, where good  $G_{r}$ is exclusive to this submarket.

\begin{table}[!h]
\caption{Closed submarket $\CM_r$ for $r^{th}$ relation $p_a=p_b$}\label{tab.eq}
\begin{center} 
\begin{tabular}{|ll|}
\hline
$\CM_{EQ}$: & 2 Agents $(A_{1}, A_{2})$ and 3 Goods $(G_a, G_b, G_{r})$ {\scriptsize // $G_{r}$: an exclusive
good}\Y\M\\
& $A_{1}$: $W_{1}=(0,1,1)$ and $U_{1}(\xx) = \min\{x_a, x_{r}\}$\M\\ 
& $A_{2}$: $W_{2}=(1,0,1)$ and $U_{2}(\xx) = \min\{x_b, x_{r}\}$\M\\ 
\hline 
\end{tabular} 
\end{center} 
\end{table}

In $\CM_r$, the endowment vector $W_{i}$'s should be interpreted as (amount of $G_a$, amount of $G_b$, amount of
$G_{r}$), i.e., in the same order of goods as listed on the first line of the table; we use similar representation in the
subsequent constructions.

\begin{lemma}\label{lem.eq}
Consider the market $\CM_r$ of Table \ref{tab.eq}.
\begin{itemize}
\item $\CM_r$ is a closed submarket. 
\item At equilibrium, $\CM_r$ enforces $p_a = p_b$.
\item Every non-negative solution of $p_a = p_b$ gives an equilibrium of $\CM_r$. 
\end{itemize} 
\end{lemma} 

\begin{proof}
Let $\alpha$ and $\beta$ denote the utility obtained by $A_{1}$ and $A_{2}$ at equilibrium respectively. Then using (\ref{eq.leon})
which characterizes optimal bundles for Leontief functions, the market clearing conditions of the two agents give: 
\[ p_b + p_{r} = \alpha(p_a + p_{r}) \ \ \ \mbox{and} \ \ \ p_a + p_{r} = \beta(p_b + p_{r}) \enspace . \] 

Clearly, the above conditions imply that $\alpha\beta=1\Rightarrow \beta=\nfrac{1}{\alpha}$.
Note that $A_{1}$ and $A_{2}$ consume $\alpha$ and $\beta$ amounts of good $G_{r}$ respectively. And since this good is
exclusive to $\CM_r$, no other agent will consume it.
Further, there are exactly two units of $G_{r}$ available in the entire market $\CM$. Hence we get, 
\[ \alpha + \beta \le 2\enspace . \]

Replacing $\beta=\frac{1}{\alpha}$ in the above condition gives $(\alpha-1)^2\le 0 \Rightarrow \alpha=\beta=1$.
Therefore, we get that every equilibrium of $\CM_r$ enforces 
$p_a+p_{r}=p_b+p_{r} \Rightarrow p_a=p_b$. Further, $\CM_r$ is a closed submarket because at equilibrium, demand of every good {\em in
$\CM_r$} is equal to its supply {\em in $\CM_r$} even though every good except $G_{r}$ might participate in the rest of the
market as well. 

For the last part, if $p_a = p_b \ge 0$, then choosing $p_{r}=1$, and $x_{1a}=x_{1r}=x_{2b}=x_{2r}=1$ gives a
market equilibrium of $\CM_r$.
\end{proof}

\subsubsection{Submarket for relation $(LIN.)\ p_a = Bp_b+ Cp_c+Dp_s$}\label{sec.lin}

The gadget for $(LIN.)$ is an extension of $(EQ.)$ having two agents with Leontief utility functions, as given in Table
\ref{tab.lin}, where $B,C,D\ge 0$.

\begin{remark}
For simplicity, we denote agents of each submarket by $A_1, A_2, \cdots$, and sometimes exclusive goods by $G_1, G_2,
\cdots$, however they are different across submarkets.  
\end{remark}
\begin{table}[!h]
\caption{$\CM_r$: Closed market for $r^{th}$ relation $p_a = Bp_b+ Cp_c+Dp_s,\ B,C,D\ge0$}\label{tab.lin}
\begin{center} 
\begin{tabular}{|ll|}
\hline
$\CM_r$: & 2 Agents $(A_{1}, A_{2})$ and 5 Goods $(G_a, G_b, G_c, G_s, G_{r})$ {\scriptsize // $G_{r}$: an exclusive
good}\Y\M\\
& $A_{1}$: $W_{1}=(1,0,0,0,1)$ and $U_{1}(\xx) = \min\{\frac{x_b}{B}, \frac{x_c}{C}, \frac{x_s}{D}, x_{r}\}$\M\\ 
& $A_{2}$: $W_{2}=(0,B,C,D,1)$ and $U_{2}(\xx) = \min\{x_a, x_{r}\}$\M\\ 
\hline 
\end{tabular} 
\end{center} 
\end{table}

\begin{lemma}\label{lem.lin}
Consider the market $\CM_r$ of Table \ref{tab.lin} with $B,C,D\ge 0$.
\begin{itemize}
\item $\CM_r$ is a closed submarket. 
\item At equilibrium, $\CM_r$ enforces $p_a = B p_b + Cp_c +D p_s$.
\item Every non-negative solution of $p_a = Bp_b + Cp_c+Dp_s$ gives an equilibrium of $\CM_r$. 
\end{itemize} 
\end{lemma} 

\begin{proof}
The proof is similar to that of Lemma \ref{lem.eq}. 
Let $\alpha$ and $\beta$ denote the utility obtained by $A_{1}$ and $A_{2}$ at equilibrium respectively. Then using (\ref{eq.leon})
together with market clearing conditions of the two agents, we get: 

\[ p_a + p_{r} = \alpha(Bp_b + Cp_c + Dp_s+ p_{r}) \ \mbox{ and }\  
 Bp_b + Cp_c + Dp_s+p_{r} = \beta(p_a + p_{r})\enspace . \] 

Clearly, the above conditions imply $\alpha\beta=1\Rightarrow \beta=\nfrac{1}{\alpha}$. Since $G_{r}$ is exclusive to this market, 
using similar argument as the proof of Lemma \ref{lem.eq} we get that $\alpha +\beta\le 2$. This together with
$\beta=\nfrac{1}{\alpha}$ gives $\alpha=\beta=1$.
Thus every equilibrium of $\CM_r$ enforces $p_a +p_{r}= Bp_b + Cp_c+Dp_s+p_{r} \Rightarrow p_a = Bp_b +
Cp_c+Dp_s$. Hence, $\CM_r$ is a closed submarket.

For the last part, if $p_a = Bp_B + Cp_c+Dp_s \ge0$, then setting $p_{r}=1$, $x_{1r}=x_{2r}=x_{2a}=1$, and 
$x_{1b}=B, x_{1,c}=C, x_{1s}=D$ gives a market equilibrium of $\CM_r$.
\end{proof} 

Using Lemma \ref{lem.lin}, we easily get the following. 

\begin{corollary}
There is a simple closed submarket to establish any linear relation of form $p_a = E_1p_{b_1} + \dots + E_np_{b_n}+E_0 p_s$ for any
$n\ge 1$, where $E_0,E_1, \dots, E_n$ are non-negative rational constants. 
\end{corollary} 

\subsubsection{Submarket for relation $(QD.)\ p_a = \frac{p_bp_c}{p_s}$}\label{sec.qd}

In this section, we derive a closed submarket for establishing the $(QD.)$ relation. In order to simplify the market construction,
which is quite involved, we first make the following two assumptions, which are removed later.  First, that $p_s = 1$ and second, that
$p_b \neq 0$. The first assumption violates the scale invariance of prices, see Section \ref{sec.scale}, but simplifies the relation needed
to $p_a = p_b p_c$. The second assumption ensures that no agent can demand an infinite amount of a good of price $p_b$ (Note that in
the reduction, since the price of a good corresponds to the probability of playing a certain strategy, eventually we do need to allow
for $p_b = 0$.).

The main idea for enforcing simpler relation, $p_a = p_b p_c$, is to ensure that there is an agent $A$ whose initial endowment is
one unit of Good 1 priced at $p_a$, and she desires to consume only Good $2$ priced at $p_b$. The left over amount of Good $2$ after
everyone, except agent $A$, consume is exactly $p_c$. Since $p_b>0$, agent $A$ has to buy all of this left over amount which requires
her to spend $p_b p_c$. On the other hand her earning from the sell of Good $1$ is $p_a$, implying $p_a=p_b p_c$. Figure
\ref{fig.sample} illustrates the idea.  The difficulty in implementing this idea lies in the fact that
$p_b$ and $p_c$ are variables; if they were both constants, the construction of the submarket would have been easy. 

\begin{figure}[htbp]
   \centering
  \includegraphics[width=0.4\textwidth]{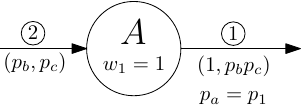}
 \caption{The main idea for enforcing relation $p_a=p_bp_c$. Wires are numbered in circle, and wire $i$ carries
 good $G_i$.  The tuple on each wire represents (amount, price).}\label{fig.sample}
\end{figure}   

In order to present the submarket in a modular manner, we will first define some devices.  Each of these devices will be implemented
via a set of agents with Leontief utility functions. Each device ensures a certain relationship between the net endowment left over by
these agents and the net consumption of these agents; for convenience, we will call these the {\em net endowment and net consumption of
the device}. Clearly, at equilibrium prices, for each device, the total worth of its
net endowment and net consumption must be equal.  
\medskip

{\bf Submarkets for the devices.} In this section, we show implementation of three devices to be used in the submarket for $(QD.)$ relation. 
\medskip

{\bf Converter (Conv($q$)):}  
The net consumption of this device is one unit of good $G_1$, whose price is $p$, and the net endowment is $p/q$ units of
good $G_2$, whose price is $q$. 
Table \ref{tab.conv} and Figure \ref{fig.conv} illustrate the implementation. In the figure tuple on edges represent {\em
(amount, price)} of the goods whose number is shown in circle. Table \ref{tab.conv} has two parts: Part 1 describes the market
and Part 2 enforces linear relations among prices using the submarkets described in Sections \ref{sec.eq} and \ref{sec.lin}. 

\begin{figure}[!htbp]
   \centering
  \includegraphics[width=0.7\textwidth]{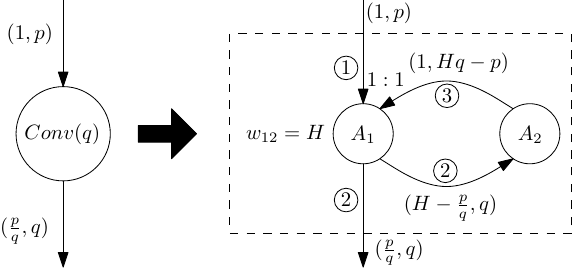}
 \caption{Flow of goods in Part 1 of Table \ref{tab.conv} for Conv($q$). Wires are numbered in circle, and wire $i$ carries
 good $G_i$.  The tuple on each wire represents (amount, price).}\label{fig.conv}
\end{figure}

\begin{table}[!ht]
\caption{A closed submarket for Conv($q$)}
\begin{center} 
\begin{tabular}{|ll|}
\hline
\multirow{5}{*}{\bf \color{blue}Part 1:} & {\bf Input:} $1$ unit of $G_1$ at price $p$ \Z\\
& {\bf Output:} $p/q$ units of $G_2$ at price $q$\\
& 2 Agents $(A_{1}, A_{2})$, 3 goods $(G_1, G_2, G_3)$\\ 
& $A_{1}$: $W_{1 2}=H$ and $U_1(\xx) = \min\{x_{1}, x_3\}$\\ 
& $A_{2}$: $W_{2 3}=1$ and $U_2(\xx)=x_{2}$ \B\\ 
\hline
\multirow{3}{*}{\bf \color{blue}Part 2:} & Closed submarkets for the following linear relations \Z\\
& $p_{2} = q$\\
& $p_{3} = Hq-p$ \B\\ 
\hline 
\end{tabular} 
\label{tab.conv}
\end{center} 
\end{table}

There are two agents $A_1$ and $A_2$, and three goods $G_1, G_2$ and $G_3$. The endowment of $A_1$ is $H$ units of $G_2$,
whose price is set to $q$. Recall that $H$ is a constant defined in Section \ref{sec.scale}. $A_1$ likes to consume $G_1$
and $G_3$ in the ratio of 1:1. The net consumption of this device, i.e., one unit of $G_1$ at price $p$, is consumed by
$A_1$. Agent $A_2$'s endowment is one unit of $G_3$, whose price is set to $Hq - p$.  $A_2$ wants to consume $G_2$,
whose price is $q$. Hence, it consumes $H - p/q$ units of $G_2$ (observe that there is no need to perform the division
involved in $p/q$ explicitly). The remaining $p/q$ units of $G_2$ form the net endowment of the device, as required. 
\medskip

{\bf Combiner (Comb($l,p_a,p_b$)):}  
The net consumption of this device is $l$ units each of goods $G_1$ and $G_2$, whose prices are $p_a$ and $p_b$, respectively. The
net endowment is $l$ units of a good $G_3$, whose price is $p_a + p_b$. Table \ref{tab.comb} and Figure \ref{fig.comb}
illustrate the implementation.

\begin{figure}[!htbp]
   \centering
  \includegraphics[width=\textwidth]{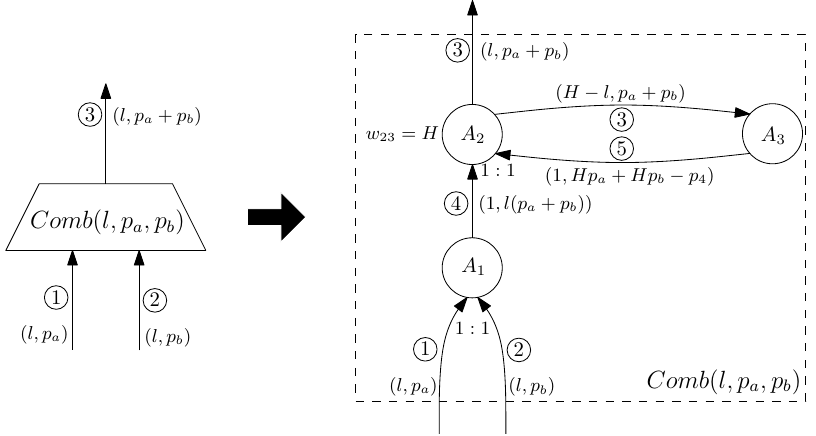}
 \caption{Flow of goods in Part 1 of Table \ref{tab.comb} for Comb($l,p_a,p_b$). Wires are numbered, and wire $i$ carries good $G_i$. 
 The tuple on each wire represents (amount, price).}\label{fig.comb}
\end{figure}

\begin{table}[!hb]
\caption{A closed submarket for Comb($l,p_a,p_b$)}
\begin{center} 
\begin{tabular}{|ll|}
\hline
\multirow{6}{*}{\bf \color{blue}Part 1:} & {\bf Input:} $l$ units of $G_1$ and $G_2$ at price $p_a$ and $p_b$ respectively \Z\\
& {\bf Output:} $l$ units of $G_3$ at price $p_a+p_b$\\
& 3 Agents $(A_{1}, A_{2}, A_3)$, 5 goods $(G_1, G_2, G_3, G_4, G_5)$\\ 
& $A_{1}$: $W_{1 4}=1$ and $U_1(\xx) = \min\{x_{1}, x_2\}$\\ 
& $A_{2}$: $W_{2 3}=H$ and $U_2(\xx) = \min\{x_{4}, x_5\}$\\ 
& $A_{3}$: $W_{3 5}=1$ and $U_2(\xx)=x_{3}$\B\\ 
\hline
\multirow{3}{*}{\bf \color{blue}Part 2:} & Closed submarkets for the following linear relations \Z\\
& $p_{3} = p_a+p_b$ \\ 
& $p_{5} = Hp_a+Hp_b-p_4$ \B\\ 
\hline 
\end{tabular} 
\label{tab.comb}
\end{center} 
\end{table}

Agent $A_1$ wants $G_1$ and $G_2$ in the ratio 1:1, and no other agent wants these goods. Therefore, $A_1$
will consume all of the available $G_1$ and $G_2$ and hence the price of her endowment, i.e., one unit of $G_4$, will be
$l(p_a + p_b)$ (observe that there is a multiplication involved in this price; however, it is not performed explicitly). 

Agent $A_2$ wants $G_4$ and $G_5$ in the ratio 1:1. The price of $A_3$'s endowment, i.e., one unit of $G_5$ is set to 
$H(p_a + p_b) - p_4$. Hence the endowment of $A_2$, i.e., $H$ units of $G_3$, has a price of $(p_a + p_b)$. Of this, $A_3$
must consume $(H - l)$, leaving $l$ amount of $G_3$ as the net endowment of this device.
\medskip

{\bf Splitter (Spl($l,p_a,p_b$)):}  
The net endowment of this device is $l$ units each of two goods $G_2$ and $G_3$, whose prices are $p_a$ and $p_b$, respectively. The
net consumption is $l$ units of Good $1$, whose price is $p_a + p_b$. Table \ref{tab.spl} and Figure \ref{tab.spl} illustrate
the implementation.

\begin{figure}[!htbp]
   \centering
  \includegraphics[width=\textwidth]{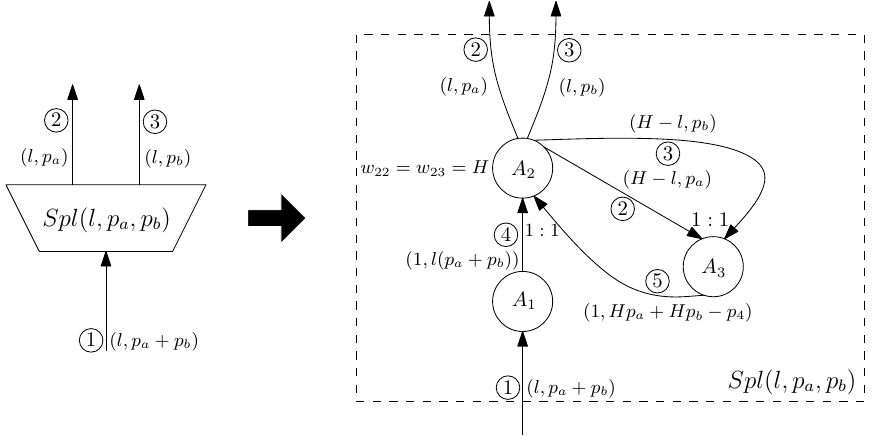}
 \caption{Flow of goods in Part 1 of Table \ref{tab.spl} for Spl($l,p_a,p_b$). Wires are numbered, and wire $i$ carries good $G_i$. 
 The tuple on each wire represents (amount, price).}\label{fig.spl}
\end{figure}

\begin{table}[!h]
\caption{A closed submarket for Spl($l,p_a,p_b$)}
\begin{center} 
\begin{tabular}{|ll|}
\hline
\multirow{6}{*}{\bf \color{blue}Part 1:} & {\bf Input:} $l$ units of $G_1$ at price $p_a+p_b$ \Z\\
& {\bf Output:} $l$ units of $G_2$ and $G_3$ at price $p_a$ and $p_b$ respectively\\
& 3 Agents $(A_{1}, A_{2}, A_3)$, 5 goods $(G_1, G_2, G_3, G_4, G_5)$\\ 
& $A_{1}$: $W_{1 4}=1$ and $U_1(\xx) = x_{1}$\\ 
& $A_{2}$: $W_{2 2}=W_{2 3}=H$ and $U_2(\xx) = \min\{x_{4}, x_5\}$\\ 
& $A_{3}$: $W_{3 5}=1$ and $U_2(\xx)=\min\{x_2,x_{3}\}$ \B\\ 
\hline
\multirow{2}{*}{\bf \color{blue}Part 2:} & Closed submarkets for the following linear relation \Z\\
& $p_2=p_a$ \\
& $p_3=p_b$ \\
& $p_{5} = Hp_a+Hp_b-p_4$ \B\\ 
\hline 
\end{tabular} 
\label{tab.spl}
\end{center} 
\end{table}

Good $G_1$ is desired only by Agent $A_1$. Hence,
the price of her initial endowment, i.e., one unit of $G_4$, is forced to be $l(p_a + p_b)$ (observe that the multiplication
involved is not done explicitly). Agent $A_2$ wants goods $G_4$ and $G_5$ in the ratio 1:1. The price of $G_5$ is set
explicitly to $H(p_a + p_b) - p_4$. The endowment of $A_2$ is $H$ units each of $G_2$ and $G_3$, whose prices have been set
to $p_a$ and $p_b$, respectively. Agent $A_3$ wants these two goods in the ratio 1:1, and because of the setting of the
price of her initial endowment, she must consume $(H-l)$ units of each of these two goods. The remaining amounts, i.e., $l$
each, form the net endowment of the device, as required.
\medskip

{\bf Submarket construction for $p_a = {p_bp_c}$.} Now we are ready to describe a closed submarket that enforces $p_a = {p_bp_c}$.
Consider the submarket given in Table \ref{tab.qd1}. 
In this market, the 7 goods, $G_1, \ldots G_7$ are exclusive to the submarket; the price of good $G_j$ is $p_j$.
The prices of $G_1, G_2, G_4, G_5, G_6, G_7$ are set to appropriate using $(EQ.)$ and $(LIN.)$ relations and prices $p_a, p_b$
and $p_c$, as specified in the second part of the table. 
The submarket uses two Converters, one Combiner and one Splitter. Each of these devices is specified by giving its (net endowment, net
consumption).  Besides the agents needed to implement these devices, the submarket requires two additional agents, $A_1$ and
$A_2$.

\begin{table}[!h]
\caption{A closed submarket $\CM'_r$ that enforces $p_a = {p_bp_c}$}
\begin{center} 
\begin{tabular}{|ll|}
\hline
\multirow{7}{*}{\bf \color{blue}Part 1:} & 2 Agents $(A_{1}, A_{2})$, 2 Converters $(Conv_{1}, Conv_{2})$, \Z\\
& 1 Combiner $(Comb)$, 1 Splitter $(Spl)$, and 7 Goods $(G_{1}, \dots, G_{7})$ \B\\ 
& $A_{1}$: $W_{1 1}=1$ and $U_1(\xx) = x_{4}$\M\\ 
& $A_{2}$: $W_{2 6}=1$ and $U_2(\xx)=x_{5}$\M\\ 
& $Conv_{1}= Conv(1)$: $(G_1, G_2)$ \M\\
& $Conv_{2}= Conv(p_b)$: $(G_6, G_7)$ \M\\ 
& $Comb(p_c,p_b,1)$:  $((G_2, G_7), G_3)$ \M \\ 
& $Spl(p_c,p_b,1)$:  $(G_3, (G_4,G_5))$ \M \B\\
\hline
\multirow{8}{*}{\bf \color{blue}Part 2:} & Closed submarkets for the following linear relations \Z \\
& $p_{1} = p_c$\\
& $p_{2} = 1$ \\ 
& $p_{4} = 1$ \\  
& $p_{5} = p_b$ \\
& $p_{6} = p_a$ \\
& $p_{7} = p_b$ \M\\
\hline 
\end{tabular} 
\label{tab.qd1}
\end{center} 
\end{table}

\begin{figure}[htbp]
   \centering
  \includegraphics[width=0.7\textwidth]{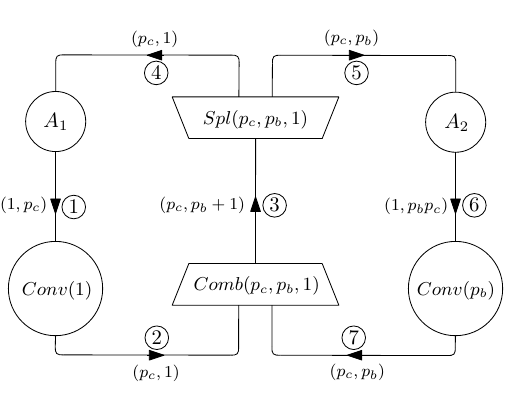}
 \caption{Flow of goods in Part 1 of Table \ref{tab.qd1}. Wires are numbered, and wire $i$ carries good $r_i$. 
 The tuple on each wire represents (amount, price).}\label{fig.qd}
\end{figure}

\begin{lemma}\label{lem.qd1}
The submarket given in Table \ref{tab.qd1} (and illustrated in Figure \ref{fig.qd}) enforces $p_a = {p_bp_c}$ and is closed at
equilibrium under the assumption $p_b \neq 0$.  
\end{lemma}
\begin{proof}
Let $\pp$ be an equilibrium price vector of the entire market, where $p_b>0$. It is easy to see that prices $p_2, p_4, p_5$ and $p_7$ are
strictly greater than zero, so these have to be consumed completely. Further, since the $(EQ.)$ and $(LIN.)$ submarkets implementing the
second part of Table \ref{tab.qd1} are closed (Lemmas \ref{lem.eq} and \ref{lem.lin}), the net endowment of goods $G_1,\dots,G_7$
for agents in this submarket, including those in devices, is exactly what they bring. 

We have set $p_1 = p_c$ and the endowment of $A_1$ is one unit of $G_1$. Good $G_1$ is desired only by the first agent in
$Conv_1$;
hence its net consumption costs $p_c$. Furthermore, since the price of $G_2$ is set to $1$ and the parameter of $Conv_1$
is $1$, the net endowment of $Conv_1$ will be $p_c$ units of $G_2$. 

The first agent in $Comb$ wants $G_2$ and $G_7$ in the ratio 1:1 both of whose prices are positive.  Moreover, net
endowment $p_c$ of Good $G_2$ is desired only by this agent, and therefore using (\ref{eq.leon}) at equilibrium the agent
has to consume $p_c$ units of both the goods.  Since prices of $G_2$ and $G_7$ are $1$ and $p_b$ respectively, 
the price of the net endowment of $Comb$ will be $p_b + 1$, and the amount will be $p_c$. 

Thus, output of $Comb$ is $p_c$ amount of good $G_3$ priced at $p_b+1$, which has to be consumed by $Spl$.
The prices $p_4$ and $p_5$ are set to $1$ and $p_b$, respectively, thereby ensuring that the total worth of the net
endowment and net consumption of $Spl$ are equal. Finally, agent $A_2$ gets $p_c$ amount of $G_5$ whose price is $p_b$
and has one unit of $G_6$ as her endowment. Hence the price of $G_6$ must be $p_6 = p_a = p_b p_c$, as required. Good $G_6$
is only desired by the first agent in $Conv_2$. The net endowment of this device is $p_c$ units of $G_7$ whose price is set
to $p_b$, and this good is fully consumed by the first agent of $Comb$. 

Note that it may be possible that $p_c$ is zero which may force prices of some of $G_1,\dots,G_7$ goods to be zero, like
$G_1$. However, whoever consumes this good also want to consume another good with non-zero price, in the same proportion.
And therefore demand of no good will exceed the supply.  Since the total supply and demand of each of the seven goods are
equal, the submarket is closed at these (equilibrium) prices.
\end{proof}

Now we will modify the construction of Table \ref{tab.qd1} in order to remove the assumption $p_b>0$ and $p_s=1$. Consider the
implementation of Table \ref{tab.qd}.

\begin{table}[!h]
\caption{A closed submarket $\CM_r$ that enforces $p_a = \frac{p_bp_c}{p_s}$}
\begin{center} 
\begin{tabular}{|ll|}
\hline
\multirow{7}{*}{\bf \color{blue}Part 1:} & 2 Agents $(A_{1}, A_{2})$, 2 Converters $(Conv_{1}, Conv_{2})$, 
\Z\\ 
& 1 Combiner $(Comb)$, 1 Splitter $(Spl)$,  and 7 Goods $(G_{1}, \dots, G_{7})$ \B \\ 
& $A_{1}$: $W_{1 1}=1$ and $U_1(\xx) = x_{4}$\M\\ 
& $A_{2}$: $W_{2 6}=1$ and $U_2(\xx)=x_{5}$\M\\ 
& $Conv_{1} = Conv(p_s)$: $(G_1, G_2)$ \M\\
& $Conv_{2} = Conv(p_b+p_s)$: $(G_6, G_7)$ \M\\ 
& $Comb(\nfrac{p_c}{p_s},p_b+p_s,p_s)$:  $((G_2, G_7), G_3)$ \M \\ 
& $Spl(\nfrac{p_c}{p_s},p_b+p_s,p_s)$:  $(G_3, (G_4,G_5))$ \M\\
\hline
\multirow{8}{*}{\bf \color{blue}Part 2:} & Closed submarkets for the following linear relations \Z \\
& $p_{1} = p_c$\\
& $p_{2} = p_s$ \\ 
& $p_{4} = p_s$ \\  
& $p_{5} = p_b+p_s$ \\
& $p_{7} = p_b+p_s$ \\
& $p_a + p_c = p_6 $\\
\hline 
\end{tabular} 
\label{tab.qd}
\end{center} 
\end{table}

\begin{lemma}\label{lem.qd}
Consider the submarket $\CM_r$ of Table \ref{tab.qd},
\begin{itemize}
\item $\CM_r$ is a closed submarket. 
\item At equilibrium, $\CM_r$ enforces $p_a= \frac{p_bp_c}{p_s}$, and $\nfrac{p_c}{p_s}\le H$.  
\item Every non-negative solution of $p_a = p_bp_c$, where $p_s>0$ and $p_c\le H$, gives an equilibrium
of $\CM_r$ with $p_s=1$.  
\end{itemize} 
\end{lemma}
\begin{proof}
The only difference between the market of Tables \ref{tab.qd} and \ref{tab.qd1} are that $p_b$ is replaced with $p_b+p_s$,
and $p_c$ with $\nfrac{p_c}{p_s}$. As $p_s>0$ (Lemma \ref{lem.sms}) we have that $p_b+p_s>0$, and $\frac{p_c}{p_s}$ is
well-defined. Further, the prices to be set in Part 2 of devices are still linear, even when $l=\nfrac{p_c}{p_s}$ in 
$Comb$ and $Spl$. Thus using Lemma \ref{lem.qd1} it follows that the market is closed and $p_6=(p_b+p_s)\frac{p_c}{p_s}$,
Then using the last linear relation enforced in Part 2 of Table \ref{tab.qd} we get $p_a=\frac{p_bp_c}{p_s}$. 

For the last part, consider a non-negative $p_a$, $p_b$ and $p_c$ such that $p_a=p_bp_c$. Set $p_s=1$, the prices of goods
$G_1, G_2, G_4, G_5, G_7$ as per Part 2 of Table \ref{tab.qd}, and $p_3=p_b+2p_s$ and $p_6=(p_b+p_s)\frac{p_c}{p_s}$.
For goods within devices, set their prices as per Part 2 of Tables \ref{tab.conv}, \ref{tab.comb} and \ref{tab.spl}
respectively. For good $G_4$ in $Comb(\frac{p_c}{p_s})$ and good $G_4$ in $Spl(\frac{p_c}{p_s})$, set their prices to 
to $(p_b+2p_s)*\frac{p_c}{p_s}$. It is easy to verify that this gives an equilibrium for market $\CM_r$ of Table
\ref{tab.qd}.
\end{proof}

\subsection{Results}
\label{sec.res}

In this section, we prove the main results using the claims established in Section \ref{sec.red}.  Given a system $F$ of
multivariate polynomials as in (\ref{eq.F}), construct an equivalent set of relations $R'(F)$ consisting of only
three types of basic relations given in (\ref{eq.trans}). Let $K$ be the number of relations in $R'(F)$. For each relation
$r \in [K]$, depending on its type, we construct a market $\CM_r$ as described in Tables \ref{tab.eq}, \ref{tab.lin} and
\ref{tab.qd}. Further, for
each $r$ of type $(QD.)$ replace the corresponding devices with agents of Tables \ref{tab.conv}, \ref{tab.comb} and \ref{tab.spl}
respectively. Combine all the $\CM_r$'s to form one market $\CM$. 
Also add the agent of (\ref{eq.sms}) in $\CM$.

Since equilibrium prices of an Arrow-Debreu market are scale invariant, {\em i.e.,} if $\pp=(p_1,\dots,p_g)$ is an equilibrium price
vector then so is $\alpha\pp,\ \forall \alpha>0$, there is no loss of generality in assuming some kind of normalization of prices, for example,
$\sum_j p_j=1$, or choosing a good to be {\em num\`{e}raire}, {\em i.e.,} fix its price to $1$.\footnote{We note that in the
reduction in \cite{EY07} from a fixed-point to an Arrow-Debreu market with algebraic excess demand function, it is assumed that
$\sum_j p_j=1$ at equilibrium.} For the latter case, we require that the 
price of the num\`{e}raire good be non-zero at equilibrium. Any good, for which an agent is non-satiated\footnote{An
agent is said to be {\em non-satiated} for good $j$ if with respect to any given bundle, she can obtain more utility by consuming additional amount of good
$j$.}, qualifies to be a num\`{e}raire.

Given an equilibrium price vector $\pp$ of $\CM$, we know that $p_s>0$ due to Lemma \ref{lem.sms}. 
Henceforth by equilibrium prices $\pp$ of $\CM$ without loss of generality we mean equilibrium prices with $p_s=1$.
In the next two lemmas we establish that the equilibria of market $\CM$ exactly capture the solutions of system $F$.

\begin{lemma}\label{lem.MtoF}
If $\pp$ is an equilibrium price vector of $\CM$, then $z_j=p_j, \forall j \in[n]$ is a solution of $F$.
\end{lemma} 
\begin{proof}
Due to Lemma \ref{lem.RF}, it is enough to show that $\pp$ is a solution of $R'(F)$. 
The submarket $\CM_r$, constructed for relation $r$ of $R'(F)$, is closed and enforces $r$ at $\pp$ (first two statements of Lemmas
\ref{lem.eq}, \ref{lem.lin}, and \ref{lem.qd}). Since $\CM$ is a union of $\CM_r$'s, $\pp$ has to satisfy each of the
relation of $R'(F)$. 
\end{proof}

Next we map solutions of $F$ to equilibria of market $\CM$.

\begin{lemma}\label{lem.FtoM}
If $\zz$ is a solution of $F$, then there exists equilibrium prices $\pp$ of market $\CM$, where $p_s=1$ and $p_j=z_j,\ \forall j
\in[n]$. 
\end{lemma} 
\begin{proof}
Using Lemma \ref{lem.RF}, we can construct a non-negative solution $\pp'$ of $R'(F)$ using $\zz$ such that $p'_s=1$, $p'_j=z_j,\
\forall j \in[n]$, and $p'_j \le H, \forall j \in [N]$. 

Construct prices $\pp$ of market $\CM$, where set $p_j=p'_j,\ \forall j \in[N]$ and $p_s=1$. 
Set $x_{ss}=1$ for agent $A_s$ of (\ref{eq.sms}). 
Last statement of Lemmas \ref{lem.eq}, \ref{lem.lin}, and \ref{lem.qd}, imply that in each $\CM_r$, $\pp$
can be extended to yield an equilibrium.  Since equilibrium in $\CM$ consists of equilibrium in each $\CM_r$ with same prices for
common goods, combining these gives an equilibrium of $\CM$.
\end{proof}

Thus establishing the strong relation between solutions of $F$ and equilibria of market $\CM$, next we prove the main theorem of the
paper which will give all the desired hardness results as corollaries.

\begin{theorem}\label{thm.main}
Equilibrium prices of market $\CM$, projected onto $(p_1,\dots,p_n)$, are in one-to-one correspondence with the solutions of
$F$. Furthermore $\CM$ can be described using polynomially many bits in $size[F]$, {\em i.e.,} $size[\CM]=poly(size[F])$.
\end{theorem} 
\begin{proof}
The first part follows using Lemmas \ref{lem.MtoF} and \ref{lem.FtoM}. For the second part, 
it is enough to show that $size[\CM]=poly(size[R'(F)])$, due to Lemma \ref{lem.RF}. Let $L$ be the size of $R'(F)$.

For each of the relation $r \in [K]$ of system $R'(F)$, we add $O(1)$ agents and goods in $\CM_r$, as described in Tables \ref{tab.eq},
\ref{tab.lin}, and \ref{tab.qd} together with Tables \ref{tab.conv}, \ref{tab.comb} and \ref{tab.spl}. 
Let $n_g$ and $n_a$ be the $\# goods$ and $\# agents$ in $\CM$, then we have $n_g=O(K)$ and $n_a=O(K)$.

Clearly, each agent of $\CM_r$ brings $O(1)$ goods and has utility for exactly those many goods, and these markets are closed
submarkets (see Definition \ref{def.csm}). The utility functions of the agents are Leontief which can be 
written as (good id, coefficient).  The endowments can be captured similarly. 

Thus the encoding of endowments and utility functions of agents in $\CM_r$ requires: (i) O($\log K$) if $r$ is of type $(EQ.)$, $(ii)$
O($\log K$ + $size(B,C,D)$) if $r$ is of type $(LIN.)$, and $(iii)$ O($\log K$+$L$) if $r$ is of type $(QD.)$ as
$size[H]=O(L)$, where $H$ is a constant defined in Section \ref{sec.scale}.

Since $\CM$ is a union of $\CM_r, \forall r \in[K]$, and the agent of (\ref{eq.sms}), the size of $\CM$ is at most $O(KL)$. 
\end{proof}

Theorem \ref{thm.main} shows that finding solutions of $F$ can be reduced to finding equilibria of an exchange market with Leontief
utility functions. 

As discussed in Section \ref{sec.nash}, the problem of computing a Nash equilibrium of a $3$-player game $\AA$ can be
formulated as finding a solution of system $F_{NE}(\AA)$ (\ref{eq.fne}) of multivariate polynomials in which variables are
bounded between $[0, 1]$ (Lemma \ref{lem.nash}). 
Note that $size[F_{NE}(\AA)]=O(size(\AA))$. Further, 
since taking projection on a set of coordinates 
is a linear function, the next theorem follows using the formulation of (\ref{eq.fne}), together with Lemma
\ref{lem.nash}, and Theorems \ref{thm.fixp} and \ref{thm.main} (see Section \ref{sec.fixp} for the reduction requirements
for class FIXP). 

\begin{theorem}\label{thm.ExLeon}
Computing an equilibrium for an Arrow-Debreu exchange market under Leontief utility functions is FIXP-hard. 
\end{theorem} 

Further, since the class of Leontief utility functions is a subclass of piecewise-linear concave (PLC) utility functions, the next corollary follows. 

\begin{corollary}\label{thm.ExPLC}
Computing an equilibrium for an Arrow-Debreu exchange market under piecewise-linear concave (PLC) utility functions is FIXP-hard. 
\end{corollary} 

Note that in Theorem \ref{thm.ExLeon} and Corollary \ref{thm.ExPLC} the resulting markets are guaranteed to have an equilibrium, since it was constructed from an instance of $3$-Nash, which always has a NE (Nash's theorem \cite{nash}). We next study the complexity of checking if an Arrow-Debreu market under Leontief (and PLC) utility admits an equilibrium. It turns out that the complexity of these questions is captured by the class $\ER$. 

Theorem \ref{thm.etr} shows that checking if a $3$-player game $\AA$ has NE within $0.5$-ball at origin in $l_\infty$ norm is
$\ER$-complete. Clearly, this problem can be reduced to finding a solution of $F_{NE}(\AA)$ of (\ref{eq.fne}) 
with upper-bound on $z_{ps}$'s changed from $1$ to $0.5$ (Lemma \ref{lem.nash}). 
If this system is reduced to a market $\CM$, then $\CM$ will have an equilibrium if and only if game $\AA$ has a NE within
$0.5$-ball at origin (Theorem \ref{thm.main}). Hence, we get the following result.

\begin{theorem}\label{thm.ExEtr}
Checking existence of an equilibrium for an Arrow-Debreu exchange market under Leontief utility functions, and under PLC utility functions, is $\ER$-hard. 
\end{theorem} 

\cite{GargVaz} gave a reduction from an exchange market $\CM$ with arbitrary concave utility functions 
to an equivalent Arrow-Debreu market $\CM'$ with firms, where utility functions of all the agents are linear.
It turns out that $\CM'$ has all the goods of $\CM$, in addition to others, and equilibrium prices of $\CM$ are in
one-to-one correspondence with the equilibrium prices of $\CM'$ projected onto the prices of common
goods. Further the production functions of $\CM'$ are precisely the utility functions in $\CM$, hence representation of
$\CM'$ is in the order of the representation of $\CM$. Therefore, this reduction together with Theorems \ref{thm.ExLeon}
and \ref{thm.ExEtr}, gives the next two results.

\begin{corollary}\label{cor.ADFixp}
The problem of computing an equilibrium for an Arrow-Debreu market under linear utility functions and Leontief production, and in turn PLC (polyhedral) production sets, is FIXP-hard. 
\end{corollary} 

\begin{corollary}\label{cor.ADEtr}
The problem of checking the existence of an equilibrium for an Arrow-Debreu market under linear utility functions and Leontief production, and in turn PLC
(polyhedral) production sets, is $\ER$-hard.  
\end{corollary} 

{\bf Remark:}  Corresponding to each FIXP-hardness result stated above, there is an artificial way of defining a sufficiency condition with respect to which each of these problems is in fact FIXP-complete. The sufficiency condition states that instance $I$ of the market satisfies the condition if and only if there is a 3-Nash instance $I'$ such that the reduction given above from 3-Nash to the market transforms $I'$ to $I$. Since all such instance $I$ of the market equilibrium problem admit equilibria, this is a valid sufficiency condition. Furthermore, there is a polynomial time algorithm for checking if $I$ satisfies this sufficiency condition, since  the reduction can be inverted in P, i.e., apply the inverse of the reduction to market instance $I$ to obtain $I'$; if $I'$ is a valid 3-Nash instance, then $I$ satisfies the sufficiency condition. However, in economics, sufficiency conditions are defined from elementary considerations, and this is not acceptable from that viewpoint.

\section{Checking Existence of Equilibrium lies in $\ER$}
\label{sec.eetr}

Using the nonlinear complementarity problem (NCP) formulation of~\cite{GargMV16} to capture equilibria of PLC markets, in this section we show that checking for existence of equilibrium in PLC markets is in $\ER$, and therefore $\ER$-complete using Corollary \ref{cor.ADEtr}. For the sake of completeness next we present the NCP formulation derived in \cite{GargMV16}.

Recall the PLC utility functions and PLC production sets defined in Sections \ref{sec.util} and \ref{sec.proddef} respectively.  Using the optimal bundle and optimal production plan conditions at equilibrium for such a market, \cite{GargMV16} derived the nonlinear complementarity problem (NCP) formulation AD-NCP for market equilibrium as shown in Table \ref{AD-NCP}, and showed the Lemma \ref{lem.ad-ncp}. 
Let $\perp$ denote a complementarity constraint between the inequality and the variable (e.g., $\sum_j D^k_{fj}x^{s}_{fj} \le \sum_j C^k_{fj}x^{r}_{fj} + T^k_f \perp \delta^k_f \ge 0$ is a shorthand for $\sum_j D^k_{fj}x^{s}_{fj} \le \sum_j C^k_{fj}x^{r}_{fj} + T^k_f;\ \ \delta^k_f\ge 0; \ \ \delta^k_f(\sum_j D^k_{fj}x^{s}_{fj} - \sum_j C^k_{fj}x^{r}_{fj} - T^k_f)=0$).

\begin{table}[!h]
\caption{AD-NCP}\label{AD-NCP}
\vskip -0.5cm
\begin{eqnarray*}
\forall (f,k):\ \sum_j D^k_{fj}x^{s}_{fj} \le \sum_j C^k_{fj}x^{r}_{fj} + T^k_f &\ \ \ \perp \ \ \ & \delta^k_f\ge 0\\
\forall (f,j):\ p_j \le \sum_k D^k_{fj}\delta^k_f & \ \ \ \perp \ \ \ & x^{s}_{fj}\ge 0 \\ 
\forall (f,j):\ \sum_k C^k_{fj}\delta^k_f \le p_j & \ \ \ \perp \ \ \ & x^{r}_{fj}\ge 0 \\
\forall (i,j):\ \sum_k U^k_{ij}\gamma^k_{i} \le \lambda_i p_j & \ \ \ \perp \ \ \ & x_{ij}\ge 0 \\
\forall (i,k):\ u_i \le \sum_j U^k_{ij}x_{ij} + T^k_{i} &\ \ \ \perp \ \ \ & \gamma^k_{i}\ge 0 \\
\forall i:\ \sum_j x_{ij}p_j \le \sum_j W_{ij}p_j + \sum_f \Theta_{if}\phi_f & \ \ \ \perp \ \ \  &\lambda_i\ge 0 \\
\forall j:\ \sum_i x_{ij} + \sum_f x^{r}_{fj} \le 1 + \sum_f x^{s}_{fj} & \ \ \ \perp \ \ \ & p_j\ge 0 \\
\forall i:\ \sum_k \gamma^k_{i} = 1 & ; & \forall f:\ \phi_f = \sum_k\delta^k_fT^k_f \\ 
u_i = \li(\sum_j W_{ij} p_j + \sum_f \Theta_{if}\phi_f) +\sum_k \gamma^k_i T^k_i & ; & \sum_j p_j = 1  
\end{eqnarray*}
\end{table}

\begin{lemma}\label{lem.ad-ncp}\cite{GargMV16}
If $(\pp,\xx, \xx^{s}, \xx^{r},\plambda,\pgamma, \pdelta)$ is a solution of AD-NCP of Table~\ref{AD-NCP}, then  $(\pp,\xx,\xx^{s}, \xx^{r})$ is a market equilibrium. Further, if $(\pp,\xx,\xx^{s}, \xx^{r})$ is a market equilibrium, then $\exists (\plambda, \pgamma,\pdelta)$ such that $(\pp, \xx, \xx^{s}, \xx^{r}, $ $ \plambda, \pgamma, \pdelta)$ is a solution of AD-NCP. 
\end{lemma}

Due to Lemma \ref{lem.ad-ncp}, checking
if the market has an equilibrium is equivalent to checking if AD-NCP admits a solution. 
Since all the inequalities and equalities in AD-NCP are polynomial, and all the coefficients in these polynomials are rational
numbers, AD-NCP can be represented using signature $\{0, 1, -1, +,*, <,\le, =\}$. 
The denominators of the coefficients can be
removed by taking least common multiple (LCM) while keeping the size of coefficients polynomial in the original size.
Therefore, checking if AD-NCP has a solution can be formulated in $\ER$ (see Section \ref{sec.etr} for definition), and we get the
following result using Theorem \ref{thm.ExEtr}.

\begin{theorem}\label{thm.ExEtrComp}
Checking existence of an equilibrium in an exchange market with piecewise linear concave utility functions is $\ER$-complete. 
\end{theorem} 

The next result follows using Corollary \ref{cor.ADEtr}.

\begin{theorem}\label{thm.ADEtrComp}
Checking existence of an equilibrium in an Arrow-Debreu market with PLC utility functions and PLC (polyhedral) production sets
is $\ER$-complete.  
\end{theorem} 

\section{Leontief Utilities under Constant Number of Agents}
\label{sec.const}

In this section, we show that there is a polynomial time algorithm for finding an equilibrium in Arrow-Debreu exchange markets under Leontief utility functions provided the number of agents is a constant. This settles part of an open problem of \cite{DK08}. Consider a Leontief exchange market with $n$ goods and $d$ agents, where $d$ is a constant. The Leontief utility function of agent $i$ is given by
\[U_i(\x_i) = \min_{j\in[n]} \left\{\frac{x_{ij}}{A_{ij}}\right\},\]
where $A_{ij}\ge 0$ is the fraction of good $j$ that agent $i$ wants. Let $W_{ij}$ be the amount of good $j$ agent $i$ owns. We assume wlog that $\sum_i W_{ij}=1,\ \forall j$. Let us capture the equilibrium utility of agent $i$ in variable $\beta_i$, then the optimal bundle condition gives,
\begin{equation}\label{eq.cob}
x_{ij} = A_{ij} \beta_i 
\end{equation}
at an equilibrium. Further, if $(p_1,\dots,p_n)$ are corresponding equilibrium prices, then the market clearing conditions can be
written as, 
\begin{eqnarray}\label{eq.cmc}
\forall j\in[n], & \ \ \sum_i x_{ij} = \sum_i A_{ij}\beta_i \le 1 \ \ \ \text{ and }  \ \ \ \text{if } p_j >0 \text{ then } \sum_i A_{ij}\beta_i = 1\enspace .
\end{eqnarray}

Further, since Leontief utility function is non-satiated (given any bundle, there exists another bundle where utility increases), the agents will spend all of their earned money. 
This gives the following relation in $\bbeta$ and $\p$:
\[\sum_j p_j x_{ij} = \sum_j W_{ij} p_j \ \ \Rightarrow \ \ \beta_i = \frac{\sum_j W_{ij}p_j}{\sum_j A_{ij}p_j}, \ \forall i\in[d]\enspace .\]

 First, we show that if there is an equilibrium, then there is one where at most $d$ prices are non-zero. 

\begin{lemma}\label{lem:cnzp}
If an exchange market with Leontief utilities has an equilibrium, then there is one where at most $d$ goods have non-zero prices. 
\end{lemma}

\begin{proof}
Let $(\bbeta^*, \p^*)$ be an equilibrium. Suppose $k>d$ number of goods have non-zero prices. Wlog, we assume that the first $k$ goods have non-zero prices. Consider the following linear system in variables $\q$:
\begin{eqnarray*}
\beta^*_i \sum_{j} A_{ij}q_j  = \sum_j W_{ij}q_j, & \ \ \forall i\in [d]\\
q_j  \ge  0, & \ \ \forall j \in [k] \\
q_j = 0, &\ \ \forall j > k\enspace .  
\end{eqnarray*}

The equilibrium allocation corresponding to $(\bbeta^*, \p^*)$ is $x^*_{ij} = A_{ij} \beta^*_i$.  Note that any $\q$ satisfying the above system will give an equilibrium prices together with allocation $\xx^*$.  Further, this system is non-empty because $\p^*$ is a feasible solution. At a vertex of the above system, at least $n$ inequalities are tight. Out of these first gives at most $d$, and therefore at least $n-d$ has to be of second type implying so many $q_j$'s have to be zero.  Hence, at such a point at most $d$ $q_j$'s will be non-zero. Since it is an equilibrium, the lemma follows.
\end{proof}

Due to Lemma \ref{lem:cnzp}, to find an equilibrium, it suffices to check for every set $S$ of $d$ goods if there is an equilibrium by setting the prices of goods outside $S$ to zero. And, this can be achieved by checking the feasibility of the following system, where $\beta_i$'s and $p_j$'s are variables, 
\begin{equation}\label{eq.cfs}
\begin{aligned}
\forall j\notin S,\ p_j=0; & \ \ \ \ \forall j \in S,\ p_j\ge 0  \\
\forall j\notin S,\ \sum_i A_{ij} \beta_i \le 1; & \ \ \ \ \forall j \in S,\ \sum_i A_{ij} \beta_i = 1\\
\forall i,\ \beta_i = \frac{\sum_j W_{ij}p_j}{\sum_j A_{ij}p_j} & \enspace . 
\end{aligned}
\end{equation}

\begin{lemma}\label{lem:cfs}
If $\exists \bbeta^*, \ \pp^*$ satisfying (\ref{eq.cfs}) then they constitute an equilibrium. 
\end{lemma}
\begin{proof}
Let $x^*_{ij}=\beta^*_i A_{ij}$, then for every agent $i$ we get 
\[
\sum_{j} x^*_{ij} p^*_{j}= \beta^*_i \sum_j A_{ij} p^*_j =\sum_j W_{ij} p^*_j \ \ (\mbox{using third condition of (\ref{eq.cfs})})\enspace . 
\]

This together with the fact that $x^*_{ij}=\beta^*_i A_{ij}$ it follows that $\xx^*_i$ is an optimal bundle of agent $i$ at prices
$\pp^*$. Market clearing for goods follows from the first two conditions of (\ref{eq.cfs}).
\end{proof}

Note that system (\ref{eq.cfs}) remains unchanged if we remove price variables that are set to zero. Then it will have $2d$ variables,
The first two conditions are linear in these variables, while the third condition is of degree two. 
Since $d$ is a constant, checking non-emptiness of (\ref{eq.cfs}) can be done in polynomial time \cite{bpr,bpr-book,DK08}. 

If (\ref{eq.cfs}) turns out to be non-empty then by Lemma \ref{lem:cfs} we get an equilibrium.  Lemma \ref{lem:cnzp} implies that we need to check feasibility of this system for every subset of goods of size $d$. There are at most ${n \choose d} \le n^d$ such systems need to be checked, which is a polynomial in number because $d$ is a constant. Therefore, overall we can find an equilibrium in polynomial time, and the next theorem follows:

\begin{theorem}
Consider an Arrow-Debreu exchange market under Leontief utility functions in which the number of agents is a constant.
Then, in polynomial time we can determine if an equilibrium exists, and if so, we can find one. 
\end{theorem}

\section{Discussion}
\label{sec.discuss}

Is computing an equilibrium for a Fisher market under PLC utilities FIXP-hard? Clearly, the problem is in FIXP since Fisher markets are a subcase of
Arrow-Debreu markets. We believe that existing techniques, for example of \cite{VY} establishing hardness for Fisher markets under SPLC utilities via reduction from
Arrow-Debreu markets, will not work and new ideas are needed. 

In economics, uniqueness of equilibria plays an important role. In this vein, we ask what is the complexity of deciding if a PLC or Leontief market has more than one equilibria. We note that the reduction given in this paper blows up the number of equilibria and hence it will not answer this question in a straightforward manner.

As stated in the Introduction, several problems which have been shown to be in FIXP, are waiting for proofs of FIXP-hardness; we give two of them below. In this context, we believe our paper has given a new, and often easier, way of obtaining such results. So far, the only recourse was to reduce from 3-Nash. However, as was done in the current paper, one could reduce from 3-Nash to simultaneous multivariate polynomial equations, hence ensuring that the instances of the latter were positive, and then to the target problem. As was the case in the current paper, the latter problem was ``closer'', than 3-Nash, to the target problem. 

The problem of finding an
approximate equilibrium under CES utilities was shown to be PPAD-complete in \cite{CPY}. This paper also showed that exact equilibrium computation is in FIXP and left the open problem of showing FIXP-hardness.

The forty-plus year-old problem of determining the computational complexity of the Hylland-Zeckhauser scheme for one-sided matching markets~\cite{HZ79} was partially resolved recently in~\cite{VY2}. They gave an irrational example and a proof of membership of the problem in FIXP, and left the open problem of proving FIXP-hardness.   
\medskip
\medskip
\medskip
\medskip

\bibliographystyle{abbrv}
\bibliography{kelly,sigproc}
\end{document}